\documentclass{article}

\usepackage{latexsym,amsmath,amssymb,url}
\usepackage{amsthm}
\usepackage{tikz}
\usetikzlibrary{decorations.pathreplacing}

 \newtheorem{lemma}{Lemma}
 \newtheorem{corollary}{Corollary}
\newtheorem{definition}{Definition}
\newtheorem{theorem}{Theorem}

\newcommand{\arc}[1]{\overrightarrow{{#1}}}

\newcommand{\multiplicity}[3]{\mu_{{#1}}(\arc{{#2}{#3}})}
\newcommand{\unDirMultiplicity}[3]{\mu_{{#1}}({#2}{#3})}
\newcommand{\torso}{\textrm{torso}}

\newcommand{\tw}{\textrm{tw}}

\newcommand{\kaCPP}{{\sc $k$-arc CPP}}
\newcommand{\BCPP}{{\sc BCPP}}

\begin{document}
\title{Parameterized Complexity of the $k$-Arc Chinese Postman Problem\thanks{A prelimary version of this paper was accepted for publication in the \emph{Proceedings of ESA 2014}.}}
\author{Gregory Gutin \and Mark Jones \and Bin Sheng \\
\small Department of Computer Science\\[-3pt]
\small  Royal Holloway, University of London\\[-3pt]
\small Egham, Surrey TW20 0EX, UK\\[-3pt]
}
\maketitle

\begin{abstract}
In the Mixed Chinese Postman Problem (MCPP), given an edge-weighted mixed graph $G$ ($G$ may have both edges and arcs), our aim is to find a minimum weight closed walk traversing each edge and arc at least once. The MCPP parameterized by the number of edges was known to be fixed-parameter tractable using a simple argument. Solving an open question of van Bevern et al., we prove that the MCPP parameterized by the number of arcs is also fixed-parameter tractable.
Our proof is more involved and, in particular, uses a well-known result of Marx, O'Sullivan and Razgon (2013)
on the treewidth of torso graphs with respect to small separators. We obtain a small cut analog of this result, and use it to construct a tree decomposition which, despite not having bounded width, has other properties allowing us to design a fixed-parameter algorithm.
\end{abstract}

\section{Introduction}

A {\em mixed graph} is a graph that may contain both edges and arcs (i.e., directed edges). A mixed graph $G$ is {\em strongly connected} if for each ordered pair $x,y$ of vertices in $G$ there is a path from $x$ to $y$ that traverses each arc in its direction. We provide further definitions and notation on (mainly) directed graphs in the next section.

In this paper, we will study the following problem.

 \begin{quote}
\fbox{~\begin{minipage}{0.9\textwidth}
  {\sc Mixed Chinese Postman Problem (MCPP)} \nopagebreak

    \emph{Instance:} A strongly connected mixed graph $G = (V, E \cup A)$, with vertex set $V$, edge set $E$ and arc set $A$; a weight function $w: E \cup A \rightarrow \mathbb{N}_0$.

    \emph{Output:} A closed walk of $G$ that traverses each edge and arc at least once, of minimum weight.
\end{minipage}~}
  \end{quote}

There is numerous literature on various algorithms and heuristics for MCPP; for informative surveys, see \cite{BeNiSoWe,Brucker1981,EiGeLa1995,Minieka1979,Peng1989}.
We call the problem the  {\sc Undirected Chinese Postman Problem (UCPP)} when  $A = \emptyset$, and the {\sc Directed Chinese Postman Problem (DCPP)} when $E = \emptyset$.
% When $A = \emptyset$, we call the problem the {\sc Undirected Chinese Postman Problem (UCPP)}, and when $E = \emptyset$, we call the problem the {\sc Directed Chinese Postman Problem (DCPP)}.
It is well-known that {\sc UCPP} is polynomial-time solvable \cite{EdJo1973} and so is {\sc DCPP} \cite{BeBo1974,Chr1973,EdJo1973},
% but {\sc MCPP} is NP-complete\footnote{Papadimitriou \cite{Papadimitriou1976} showed that {\sc MCPP} is NP-complete even when $G$ is planar with each vertex having total degree $3$ and all edges and arcs having weight $1$.} \cite{Papadimitriou1976}.
but {\sc MCPP} is NP-complete, even when $G$ is planar with each vertex having total degree $3$ and all edges and arcs having weight $1$ \cite{Papadimitriou1976}.
  It is therefore reasonable to believe that {\sc MCPP} may become easier the closer it gets to {\sc UCPP} or {\sc DCPP}.

Van Bevern {\em et al.} \cite{BeNiSoWe} considered two natural parameters
for {\sc MCPP}: the number of edges and the number of arcs.
They showed that {\sc MCPP} is fixed-parameter tractable (FPT) when parameterized by the number $k$ of edges.
That is,  MCPP can be solved in time
$f(k)n^{O(1)}$, where $f$ is a function only depending on
$k$, and $n$ is the number of vertices in $G.$ For background and terminology on parameterized complexity we refer the
reader to the monographs~\cite{DowneyFellows13,FlumGrohe06,Niedermeier06}.
Van Bevern {\em et al.}'s algorithm is as follows. Replace every undirected edge $uv$ by
either the arc $\arc{uv}$ or arc $\arc{vu}$ or the pair $\arc{uv}$
and $\arc{vu}$ (all arcs have the same weight as $uv$) and
solve the resulting {\sc DCPP}. Thus, the {\sc MCPP}
can be solved in time $O(3^kn^3)$, where $n$ is the number of
the number of vertices in $G$.
% in
% the Appendix.
% the full version of this paper \cite{GutinJonesSheng14}.
%Such an algorithm is called {\em fixed-parameter} algorithm.

%  \subsection*{Faster Algorithm for MCPP Parameterized by the Number of Edges}
  We describe a faster algorithm here.
% Let $e$ and $n$ be the number of (undirected) edges and vertices in $G$.
% The {\sc MCPP} can be solved by an $O(2^en^3)$ algorithm:
Replace every undirected edge $uv$ by the pair $\arc{uv}$ and $\arc{vu}$.
Now construct a network $N$ from the resulting digraph $D$
as follows: the cost of every arc of $G$ is the same as its
weight, the cost of every arc $\arc{xy}$ in $D$, which is not in
$G$, is the weight the undirected edge $xy$ of $G$, the lower bound of
every arc of $G$ is 1, and for each pair $\arc{uv}$ and $\arc{vu}$ of
arcs that replaced an undirected edge $uv$, we assign lower
bound 0 to one of the edges and 1 to the other. All upper bounds are
$\infty$. Find a minimum-cost circulation (i.e., a flow of value 0) in $N$.
This will correspond to a closed walk in $D$ in which all arcs of
$G$ are traversed at least once and at least one of the arcs
$\arc{uv}$ and $\arc{vu}$ corresponding to an undirected edge $uv$ of $G$ is
traversed at least once (the arc whose lower bound is 1 in $N$).
As there are $2^k$ ways to assign lower bounds to the pairs of arcs in $N$, we obtain
% the stated running time.
a running time of $O(2^kn^3)$.

% TODO QUOTE A RUNNING TIME FOR FLOWS?

Van Bevern {\em et al.} \cite{BeNiSoWe} and Sorge  \cite{Sor2013} left it
as an open question whether the {\sc MCPP} is fixed-parameter
tractable when parameterized by the number of arcs.
This is the parameterization we consider in this paper.

\begin{quote}
\fbox{~\begin{minipage}{0.9\textwidth}
  {\sc $k$-arc Chinese Postman Problem (\kaCPP)}

  \emph{Instance:} A strongly connected weighted mixed graph $G = (V, E \cup A)$, with vertex set $V$, edge set $E$ and arc set $A$; a weight function $w: E \cup A \rightarrow \mathbb{N}_0$.
    \nopagebreak

      \emph{Parameter:} $k = |A|$.

      \emph{Output:} A closed walk of $G$ that traverses each edge and arc at least once, of minimum weight.
\end{minipage}~}
  \end{quote}

This parameterized problem
is of practical interest, for example, if we view the mixed graph as a network of
streets in a city: while edges represent two-way streets, arcs are for
one-way streets. Many cities have a relatively small number of one-way
streets and so the number of arcs appears to be a good parameter for optimizing, say, police patrol
in such cities \cite{BeNiSoWe}.

We will assume for convenience that the input $G$ of {\sc $k$-arc CPP}  is a {\em simple} graph, i.e. there is at most one edge or one arc (but not both) between any pair of vertices.
 The multigraph version of the problem may be reduced to the simple graph version by subdividing arcs and edges.
As the number of arcs and edges is at most doubled by this reduction, this does not affect the parameterized complexity of the problem.

We will show that \kaCPP{} is fixed-parameter tractable. Our proof is significantly more complicated than the ones described above for the {\sc MCPP} parameterized by the number of edges.
For that problem, as we saw, we can replace the undirected edges with arcs.
However a similar approach for {\sc MCPP} parameterized by the number of arcs (replacing arcs with edges) does not work.
Instead, in FPT time, we reduce the problem to the   {\sc Balanced Chinese Postman Problem} (\BCPP), in which there are no arcs, but instead
a demand function on the imbalance of the vertices is introduced (the parameter for the {\sc BCPP} is based on the values of the demand function).
This reduction is only the first step of our proof, as unfortunately the {\sc BCPP} is still NP-hard, unlike the {\sc DCPP}.

The {\sc BCPP} turns out to be polynomial time solvable as long as a certain connectivity property holds. Solving the problem in general requires making some guesses on the edges in certain small cuts in the graph. To keep the running time fixed-parameter, we require a structure on the graph that allows us to only consider a few edges from small cuts at a time.
To acheive this, we make use of a recent result of Marx, O'Sullivan and Razgon \cite{MarxOSullivanRazgon2011} on the treewidth of torso graphs with respect to small separators.

% We will show that \kaCPP{} is fixed-parameter tractable. Our proof is significantly more complicated than the one for the {\sc MCPP} parameterized by the number of edges as a similar approach will not work.
% Instead, in FPT time, we reduce the problem to the   {\sc Balanced Chinese Postman Problem} (\BCPP), in which there are no arcs but instead
% a demand function on the imbalance of the vertices. Unfortunately this problem is still NP-hard, and so we must use further techniques to solve the problem.

Marx, O'Sullivan and Razgon \cite{MarxOSullivanRazgon2011} use the following notion of a graph torso. Let $G = (V, E)$ be a graph and $S \subseteq V.$ The graph $\torso(G,S)$ has vertex set $S$ and vertices $a,b \in S$ are connected by an edge $ab$ if $ab \in E$ or there is a path in $G$ connecting $a$ and $b$ whose internal vertices are not in $S$.

Marx {\em et al.} \cite{MarxOSullivanRazgon2011}  show that for a number
of graph separation problems, it is possible to derive a graph closely
related to a torso graph, which has the same separators as the original
input graph. The separation problem can then be solved on this new graph,
which has bounded treewidth. By contrast, we use the torso graph as a tool
to construct a tree decomposition of the original graph, which does not
have bounded width, but has enough other structural restrictions to make a
dynamic programming algorithm possible.
So, our application of Marx {\em et al.}'s result is quite different from
its use in \cite{MarxOSullivanRazgon2011}, and we believe it may be used
for designing fixed-parameter algorithms for other problems on graphs.
Note that Marx {\em et al.} are interested in small separators (i.e. sets
of vertices whose removal disconnects a graph), whereas we are interested
in small cuts (sets of edges whose removal disconnects a graph).
We therefore prove an analog of Marx {\em et al.}'s result for cuts.
% This
% necessitates
% an extra step in the construction of our tree decomposition, to ensure
% that all minimal cuts are covered by minimal separators.

% However, our proof is much more complicated than the one for the {\sc MCPP} parameterized by the number of edges. ...

The rest of the paper is organized as follows. The next section contains further terminology and notation. In Section~\ref{sec:reduction}, we reduce  \kaCPP{}  to {\sc Balanced Chinese Postman Problem} (\BCPP).
In Section~\ref{sec:troad}, we introduce and study two key notions that we use to solve \BCPP{}: $t$-roads, which witness a connectivity property of the graph that makes the {\sc BCPP} easy to solve; and small $t$-cuts,  which witness the fact that a $t$-road does not exist. In Section~\ref{sec:treedecomp}, we investigate a special tree decomposition of the input graph of \BCPP.
This decomposition is used in a dynamic programming algorithm given in Section~\ref{sec:dp}. The last section contains some conclusions and open problems.

 \section{Further Terminology and Notation}\label{sec:prelim}

 For a positive integer $p$ and an integer $q$, $q<p$, $[q,p]$ will denote the set \newline $\{q,q+1,\ldots ,p\}$ and
$[p]$ the set $[1,p].$
To avoid confusion, we denote an edge between two vertices $u,v$ as $uv$,
and an arc from $u$ to $v$ as $\arc{uv}$. %An edge $uv$ may be equivalently written as $vu$.

Although we will shortly reduce the the \kaCPP{} to a problem on undirected graphs, we will still be interested in directed graphs as a way of expressing solutions. For example, a walk which is a solution to an instance of the \kaCPP{} can be represented by a directed multigraph, with one copy of an arc $uv$ for each time the walk passes from $u$ to $v$. This motivates the following definitions.

%
% For a mixed multigraph $G$, let $D$ be a directed multigraph derived from $G$ by replacing each arc $\arc{uv}$ of $G$ with multiple copies of $\arc{uv}$ (at least one), and replacing each edge in $uv$ in $G$ with multiple copies of the arcs $\arc{uv}$ and $\arc{vu}$ (such that there is at least one copy of $\arc{uv}$ or at least one copy of $\arc{vu}$).
% Then we say $D$ is a \emph{multi-orientation} of $G$.
% %If $D$ is a multi-orientation of $G$, we say that $G$ is the \emph{undirected underlying graph} of $D$.
% If $D$ has exactly one copy $\arc{uv}$ for each arc $\arc{uv}$ in $G$, and either exactly one copy of $\arc{uv}$ or exactly one copy of $\arc{vu}$ for each edge $uv$ in $G$, we say $D$ is an \emph{orientation} of $G$.
% If $D$ is an orientation of $G$ and $G$ is undirected, we say that $G$
% is the {\em undirected version} of $D$ (if $D$ has parallel arcs then $G$ has parallel edges).
%
% For a mixed multigraph $G$, $\multiplicity{G}{u}{v}$ denotes the number of arcs of the form $\arc{uv}$ in $G$, and $\unDirMultiplicity{G}{u}{v}$ denotes the number of edges of the form $uv$.
% For a weighted graph $G$ and a multi-orientation $D$ of $G$, the \emph{weight} of $D$ is the sum of the weights of all its arcs, where the weight of an arc in $D$ is the weight of the corresponding edge or arc in $G$.

For a mixed multigraph $G$, $\multiplicity{G}{u}{v}$ denotes the number of arcs of the form $\arc{uv}$ in $G$, and $\unDirMultiplicity{G}{u}{v}$ denotes the number of edges of the form $uv$.
For a mixed multigraph $G$, let $D$ be a directed multigraph derived from $G$ by replacing each arc $\arc{uv}$ of $G$ with multiple copies of $\arc{uv}$ (at least one), and replacing each edge $uv$ in $G$ with multiple copies of the arcs $\arc{uv}$ and $\arc{vu}$ (such that there is at least one copy of $\arc{uv}$ or at least one copy of $\arc{vu}$).
Then we say $D$ is a \emph{multi-orientation} of $G$.
% Given a mixed multigraph $G$ and directed multigraph $D$ with the same vertex set $V$, we say that $D$ is a \emph{multi-orientation} of $G$ if for each $u, v \in V$, the following conditions hold :
%
% \begin{enumerate}
%  \item $\multiplicity{D}{u}{v} \ge \multiplicity{G}{u}{v}$ and $\multiplicity{D}{v}{u} \ge \multiplicity{G}{v}{u}$;
%  \item $\multiplicity{D}{u}{v} + \multiplicity{D}{v}{u} \ge \multiplicity{G}{u}{v} + \unDirMultiplicity{G}{u}{v} + \multiplicity{G}{v}{u}$;
%  \item If  $\multiplicity{G}{u}{v} + \unDirMultiplicity{G}{u}{v} = 0$ then $\multiplicity{D}{u}{v}= 0$ and
%  if $\multiplicity{G}{v}{u} + \unDirMultiplicity{G}{u}{v} = 0$ then $\multiplicity{D}{v}{u}= 0$.
% \end{enumerate}
%
% In other words, $D$ is derived from $G$ by replacing each arc $\arc{uv}$ of $G$ with multiple copies of $\arc{uv}$ (at least one), and replacing each edge in $uv$ in $G$ with multiple copies of the arcs $\arc{uv}$ and $\arc{vu}$ (such that there is at least one copy of $\arc{uv}$ or at least one copy of $\arc{vu}$).
If $D$ is a multi-orientation of $G$ and $\multiplicity{D}{u}{v} + \multiplicity{D}{v}{u} = \multiplicity{G}{u}{v} + \unDirMultiplicity{G}{u}{v} + \multiplicity{G}{v}{u}$ for each $u,v \in V$
(i.e. $D$ is derived from $G$ by keeping every arc of $G$ and replacing every edge of $G$ with a single arc),
% (i.e. $D$ has exactly one arc for each edge and arc in $G$),
we say  $D$
is an {\em orientation} of $G$. If $D$ is an orientation of $G$ and $G$ is undirected, we say that $G$
is the {\em undirected version} of $D$. {The \textit{underlying graph} $G$ of an undirected multigraph $H$ can be obtained from $H$ by deleting all but one edge among all edges between $u$ and $v$ for every pair $u,v$ of vertices of $H$.}
%\textcolor{red}{ The \textit{underlying graph} $G$ of an undirected multigraph $H$ is the simple subgraph of $H$ induced by $V(H)$. We say an undirected graph $H$ is an undirected version of a directed graph $G$ if %and only if $G$ is an orientation of $H$.}

For a simple weighted graph $G$ and a multi-orientation $D$ of $G$, the \emph{weight} of $D$ is the sum of the weights of all its arcs, where the weight of an arc in $D$ is the weight of the corresponding edge or arc in $G$.

% For a mixed multigraph $G$, let $D$ be a directed multigraph derived from $G$ by replacing each arc $\arc{uv}$ of $G$ with multiple copies of $\arc{uv}$ (at least one), and replacing each edge in $uv$ in $G$ with multiple copies of the arcs $\arc{uv}$ and $\arc{vu}$ (such that there is at least one copy of $\arc{uv}$ or at least one copy of $\arc{vu}$).
% Then we say $D$ is a \emph{multi-orientation} of $G$.
%If $D$ is a multi-orientation of $G$, we say that $G$ is the \emph{undirected underlying graph} of $D$.
% If $D$ has exactly one copy $\arc{uv}$ for each arc $\arc{uv}$ in $G$, and either exactly one copy of $\arc{uv}$ or exactly one copy of $\arc{vu}$ for each edge $uv$ in $G$, we say $D$ is an \emph{orientation} of $G$.
% If $D$ is an orientation of $G$ and $G$ is undirected, we say that $G$
% is the {\em undirected version} of $D$ (if $D$ has parallel arcs then $G$ has parallel edges).
%
% For a mixed multigraph $G$, $\multiplicity{G}{u}{v}$ denotes the number of arcs of the form $\arc{uv}$ in $G$, and $\unDirMultiplicity{G}{u}{v}$ denotes the number of edges of the form $uv$.
% For a weighted graph $G$ and a multi-orientation $D$ of $G$, the \emph{weight} of $D$ is the sum of the weights of all its arcs, where the weight of an arc in $D$ is the weight of the corresponding edge or arc in $G$.

For a directed multigraph $D=(V,A)$ and $v\in V$, $d_D^+(v)$ and $d_D^-(v)$ denote the out-degree and in-degree of $v$ in $D$, respectively. Let $t:V \rightarrow \mathbb{Z}$ be a function {and $V_{t}^{+}=\{u\in V, t(u)>0\},  V_{t}^{-}=\{u\in V, t(u)<0\}$}. We say that a vertex $u$ in $D$ is \emph{$t$-balanced} if $d_D^+(u) - d_D^-(u) = t(u)$. We say that $D$ is \emph{$t$-balanced} if every vertex is $t$-balanced. Note that if $D$ is $t$-balanced then $\sum_{v\in V}t(v)=0.$
 We say that a vertex $u$ in $D$ is \emph{balanced} if $d_D^+(u) = d_D^-(u)$, and we say that $D$ is \emph{balanced} if every vertex is balanced.

% When $t(v)=0$ for all $v \in V$, we omit $t$ and speak of {\em balanced} vertices and {\em balanced} directed multigraphs.  Let $V^+_t=\{v\in V:\ t(v)>0\}$ and $V^-_t=\{v\in V:\ t(v)<0\}.$

In directed multigraphs, all walks (in particular, paths and cycles) that we consider are directed. A directed multigraph $D$ is {\em Eulerian} if there is a closed walk of $D$ traversing every arc exactly once. It is well-known that a directed multigraph $D$ is Eulerian if and only if $D$ is balanced and the undirected version of $D$ is connected
 \cite{BaGu2009}.

%  For an undirected graph $G=(V,E)$ and vertices $a,b$ of $G$, a set $S$ of edges (vertices, respectively) is called an $(a,b)$-{\em cut} ($(a,b)$-{\em separator}, respectively) if $a$ and $b$ are in different components of $G-S$.

{ For an undirected graph $G = (V,E)$, and two vertex sets $X, Y\subseteq V(G)$, an $(X,Y)$-{\em separator} ( $(X,Y)$-{\em cut}, respectively) is a set $S\subseteq V\setminus (X\cup Y)$ (a set $S\subseteq E$, respectively) such that there is no path between vertices in $X$ and $Y$ in graph $G - S$. When $X=\{x\}$ and $Y=\{y\}$, we speak of $(x,y)$-separators and $(x,y)$-cuts.}
 
% vertices $a,b$ of $G$, a set $S$ of vertices disjoint from $\{a,b\}$ is called an $(a,b)$-{\em separator} if $a$ and $b$ are in different components of $G-S$.
% A set of edges $F$ is called an $(a,b)$-{\em cut} if $a$ and $b$ are in different components of $G - F$. For 

Observe that the following is an equivalent formulation of the { \sc $k$-arc CPP}.

  \begin{quote}
\fbox{~\begin{minipage}{0.9\textwidth}
  {\sc {\sc $k$-arc Chinese Postman Problem} (\kaCPP)} \nopagebreak

    \emph{Instance:} A strongly connected mixed graph $G = (V, E \cup A)$, with vertex set $V$, edge set $E$ and arc set $A$; weight function $w: E \cup A \rightarrow \mathbb{N}_0$.

      \nopagebreak
      \emph{Parameter:} $k = |A|$.

    \emph{Output:} A directed multigraph $D$ of minimum weight such that $D$ is a multi-orientation of $G$ and $D$ is Eulerian.
\end{minipage}~}
  \end{quote}

%
%  In this paper we show that \kaCPP{} is fixed-parameter tractable.

 \section{Reduction to Balanced CPP}\label{sec:reduction}

 Our first step is to reduce \kaCPP{} to a problem on a graph without arcs.
 Essentially, given a graph $G = (V, E \cup A)$ we will ``guess'' the number of times each arc in $A$ is traversed in an optimal solution.
 This then leaves us with a problem on $G' = (V,E)$. Rather than trying to find an Eulerian multi-orientation of $G$, we now try to find a multi-orientation of $G'$ in which the imbalance between the in- and out-degrees of each vertex depends on the guesses for the arcs in $A$ incident with that vertex.

 More formally, we will provide a Turing reduction to the following problem:

   \begin{quote}
\fbox{~\begin{minipage}{0.9\textwidth}
  {\sc {\sc Balanced Chinese Postman Problem} (\BCPP)} \nopagebreak

    \emph{Instance:} An undirected graph $G = (V,E)$; a weight function $w~:~E~\rightarrow~\mathbb{N}_0$; a demand function $t:V \rightarrow \mathbb{Z}$ such that $\sum_{v \in V} t(v) = 0$. %, with $V^+ = \{v \in V: t(v) > 0\}$ and $V^- = \{v \in V: t(v) < 0\}$.
      \nopagebreak
    \emph{Parameter:} $p = \sum_{v\in V^+_t}t(v)$.
      \nopagebreak

    \emph{Output:} A minimum weight $t$-balanced multi-orientation $D$ of $G$.
\end{minipage}~}
  \end{quote}

Henceforth, any demand function $t:V \rightarrow \mathbb{Z}$ will be such that $\sum_{v \in V} t(v) = 0$.

 Observe that when $t(v)=0$ for all $v \in V$, \BCPP{} is equivalent to UCPP.
 \BCPP{} was studied by Zaragoza Mart\'{i}nez \cite{MartinezThesis} who proved that the problem is NP-hard.
  We will reduce \kaCPP{} to \BCPP{} by guessing the number of times each arc is traversed.
In order to ensure a fixed-parameter aglorithm, we need a bound (in terms of $|A|$) on the number of guesses.
We will do this by bounding the total number of times any arc can be traversed in an optimal solution.

%  \subsection{Bounding the number of uses of arcs}
%  Given a graph $G$, vertices $u,v$ in $G$ and a multi-orientation $D$ of $G$, we let $\multiplicity{D}{u}{v}$ be the number of arcs of the form $\arc{uv}$ in $D$.
%  Given an undirected multigraph $H$ with underlying simple graph $G$, we let $\unDirMultiplicity{H}{u}{v}$ be the number of edges of the form $uv$ in $H$.

 \begin{lemma}\label{lem:arcMultiplicity}
  Let $G = (V,A \cup E)$ be a mixed graph, and let $k=|A|$.
  Then for any optimal solution $D$ to \kaCPP{} on $G$ with minimal number of arcs, we have that
  $\sum_{\arc{uv} \in A} \multiplicity{D}{u}{v} \le k^2/2 + 2k$.
 \end{lemma}
 \begin{proof}
   Let $A=A_{1}\cup A_{2}$ where $A_{1}=\{\arc{uv}:\ \arc{uv} \in A~ and~ \multiplicity{D}{u}{v}\geq 3\}$ and  $A_{2}=A\setminus A_1$. Let $|A_{1}|=p$ and $|A_{2}|=k-p=q$.

Consider an arc $\arc{uv} \in A$. Since $D$ is balanced, we have that $D$ has
$\multiplicity{D}{u}{v}$ arc-disjoint directed cycles, each containing exactly one copy of $\arc{uv}$. We claim that each such cycle must contain at least one copy of an arc in $A_{2}$. Indeed, otherwise, there is a cycle $C$ containing $\arc{uv}$ that does not contain any arc in $A_{2}$, which means that $C$ consists of arcs in $A_{1}$ and arcs corresponding to (undirected) edges in $G$. We may construct a directed multigraph $D'$ as follows: Remove from $D$ two copies of each arc in $A_{1}$ that appears in $C$, and reverse the arcs in $C$ that correspond to undirected edges in $G$. Observe that $D'$ is Eulerian and is also a multi-orientation of $G$, and so $D'$ is a solution with smaller weight than $D$ or an optimal solution with fewer arcs than $D$, contradicting the minimality of $D$.

So each of the $\multiplicity{D}{u}{v}$ cycles contains at least one copy of an arc in $A_{2}$.
 Observe that $D$ has at most $2q$ copies of arcs in $A_{2}$, and so $\multiplicity{D}{u}{v}\leq 2q$. Thus, we have $\sum_{\arc{uv} \in A} \multiplicity{D}{u}{v} = \sum_{\arc{uv} \in A_{1}} \multiplicity{D}{u}{v} + \sum_{\arc{uv} \in A_{2}} \multiplicity{D}{u}{v} \le p\cdot 2q+2q \leq 2\cdot (\frac{p+q}{2})^2+2k=k^2/2+2k$.
 \end{proof}

%  \subsection{Reduction to Generalised Vertex Balancing Problem}

 Now we may prove the following:

 \begin{lemma}\label{lem:initialReduction}

Suppose that there exists an algorithm which finds the optimal solution to an instance of \BCPP{} on $(G',w',t')$ with parameter $p$ in time $f(p) |V(G')|^{O(1)}$.
Then there exists an algorithm which finds the optimal solution to an instance of \kaCPP{} on  $(G = (V, A \cup E),w)$ with parameter $k$, which runs in time $\binom{\lfloor k^2/2+2k\rfloor}{k}\cdot f(\lfloor k^2/2+2k\rfloor) \cdot |V|^{O(1)}$.

Thus, if  \BCPP{} is FPT then so is  \kaCPP.
 \end{lemma}
\begin{proof}
 Let $(G = (V, A \cup E),w)$ be an instance of \kaCPP{}, and let $k = |A|$. Let $\kappa=\lfloor k^2/2+2k\rfloor$.
By Lemma~\ref{lem:arcMultiplicity},
  $\sum_{\arc{uv} \in A} \multiplicity{D}{u}{v} \le \kappa$ for any optimal solution $D$ to \kaCPP{} on $(G,w)$ with minimal number of arcs.

Let $G' = (V,E)$ and let $w'$ be $w$ restricted to $E$.
Given a function $\phi: A \rightarrow [\kappa]$ such that $\sum_{\arc{uv} \in A} \phi(\arc{uv}) \le \kappa$, let $t_\phi: V \rightarrow [-\kappa,\kappa]$ be the function such that $t_\phi(v) = \sum_{\arc{uv} \in A}\phi(\arc{uv}) - \sum_{\arc{vu} \in A}\phi(\arc{vu})$ for all $v \in V$.
Observe that $\sum_{v\in V^+_{t_\phi}}t_\phi(v) \le \sum_{\arc{uv} \in A} \phi(\arc{uv}) \le \kappa$, and thus \BCPP{} on $(G',w',t_\phi)$ has parameter $p_\phi \le \kappa$.

{Observe that for any given solution $D_\phi$ to \BCPP{} on $(G', w', t_\phi)$, $D_{\phi}$ is $t_{\phi}$-balanced, thus $d_{D_{\phi}}^{+}(u)-d_{D_{\phi}}^{-}(u)=t_{\phi}(u)$. If we add $\phi(\arc{uv})$ copies of each arc $\arc{uv} \in A$ to $D_\phi$, and denote the resulting graph with $D$, then graph $D$ is balanced. For any vertex $v\in V(G)$, $d_{D}^{+}(v)=d_{D_{\phi}}^{+}(v)+\sum_{\arc{vu} \in A}\phi(\arc{vu})$, $d_{D}^{-}(v)=d_{D_{\phi}}^{-}(v)+\sum_{\arc{uv} \in A}\phi(\arc{uv})$, thus $d_{D}^{+}(v)-d_{D}^{-}(v)=d_{D_{\phi}}^{+}(v)-d_{D_{\phi}}^{-}(v)+\sum_{\arc{vu} \in A}\phi(\arc{vu})-\sum_{\arc{uv} \in A}\phi(\arc{uv})=t_{\phi}(v)-t_{\phi}(v)=0$. So $D$ is a connected balanced graph (and thus Eulerian), which is also a multi-orientation of $G$, and thus is a solution to \kaCPP{} on $(G,w)$ with weight $w'(D_\phi) + \sum_{\arc{uv} \in A}\phi(\arc{uv})w(\arc{uv})$.
}

{
Furthermore for any solution $D$ to \kaCPP{} on $(G,w)$, we know that $D$ is balanced, so for any vertex $v\in V(G)$, $d_{D}^{+}(v)=d_{D}^{-}(v)$. Let $\phi(\arc{uv}) = \multiplicity{D}{u}{v}$ for each $\arc{uv}\in A$ and let $D_\phi$ be $D$ restricted to $E$. For any vertex $v\in V(G)$, we have $d_{D_{\phi}}^{+}(v)=d_{D}^{+}(v)-\sum_{\arc{vu} \in A}\phi(\arc{vu})$, $d_{D_{\phi}}^{-}(v)=d_{D}^{-}(v)-\sum_{\arc{uv} \in A}\phi(\arc{uv})$, therefore, $d_{D_{\phi}}^{+}(v)-d_{D_{\phi}}^{-}(v)=\sum_{\arc{uv} \in A}\phi(\arc{uv})-\sum_{\arc{vu} \in A}\phi(\arc{vu})=t_{\phi}(v)$.
So $D_\phi$ is a $t_{\phi}$-balanced multi-orientation of $G'$ and thus a solution to \BCPP{} on $(G', w', t_\phi)$ and $D$ has weight $w'(D_\phi) + \sum_{\arc{uv} \in A}\phi(\arc{uv})w(\arc{uv})$.}
%
% by letting $\phi(\arc{uv}) = \multiplicity{D}{u}{v}$ and $D_\phi$ be $G$ restricted to $E$, It is clear that for any solution $D$ to \kaCPP{} on $(G,w)$, there exist $\phi$, $D_\phi$ such that $D_\phi$ is a solution to \BCPP{} on $(G', w', t_\phi)$ and $D$ has weight $w'(D_\phi) + \sum_{\arc{uv} \in A}\phi(\arc{uv})w(\arc{uv})$.

% Suppose that there exists an algorithm which finds the optimal solution to an instance of \BCPP{} on $(G',w',t')$ with parameter $p$ in time $f(p)|V|^{O(1)}$.
There are at most $\binom{q}{k}$ ways of choosing positive integers $x_1, \dots, x_k$ such that $\sum_{i \in [k]}x_i \le q$. Indeed, for each $i \in [k]$ let $y_i = \sum_{j=1}^{i}x_j$. Then $y_i < y_j$ for $i < j$ and $y_i \in [q]$ for all $i$, and for any such choice of $y_1, \dots, y_k$ there is corresponding choice of $x_1, \dots, x_k$ satisfying $\sum_{i=1}^{k}x_i \le q$. Therefore the number of valid choices for $x_1, \dots, x_k$ is the number of ways of choosing $y_1,\dots, y_k$, which is the number of ways of choosing $k$ elements from a set of $q$ elements.

Therefore there are at most $\binom{\kappa}{k}$ choices for a function $\phi: A \rightarrow [\kappa]$ such that $\sum_{\arc{uv} \in A} \phi(\arc{uv}) \le \kappa$.
Each choice leads to an instance of \BCPP{} with parameter at most $\kappa$. Therefore in time $\binom{\kappa}{k}f(\kappa)\cdot |V|^{O(1)} $ we can find,  for every valid choice of $\phi$, the optimal solution $D_\phi$ to \BCPP{} on $(G', w', t_\phi)$.

It then remains to choose the function $\phi$ that minimizes $w'(D_\phi) + \sum_{\arc{uv} \in A}\phi(\arc{uv})w(\arc{uv})$, and return the graph $D_\phi$ together with $\phi(\arc{uv})$ copies of each arc $\arc{uv} \in A$.

% Thus we have that an optimal solution $D$ to $\kaCPP{}$ on $(G,w)$ is equal to the optimal solution $D_\phi$.
\end{proof}

Due to Lemma~\ref{lem:initialReduction}, we may now focus on \BCPP.

%  [Note about the difficulty of \BCPP?]

%  \section{Importance of small cuts in GVBP}
 \section{Expressing Connectivity: $t$-roads and $t$-cuts}\label{sec:troad}

 Although we will not need the result until later, now is a good time
 to prove a bound for \BCPP{} somewhat similar to that in Lemma \ref{lem:arcMultiplicity}.
%  to prove a similar bound on the number of uses of edges.

  \begin{lemma} \label{lem:edgeMultiplicity}
  Let $(G,w,t)$ be an instance of \BCPP, with $p = \sum_{v\in V^+_t}t(v)$.
  Then for any optimal solution $D$ to \BCPP{} on $(G,w,t)$ with minimal number of arcs, we have that
%   Let $G$ be a simple mixed graph with  edges $E$ and arcs $A$, and let $k=|A|$.
%   Then for any optimal solution $D$ to \kaCPP{} on $G$, we have that
  $\multiplicity{D}{u}{v} + \multiplicity{D}{v}{u} \le \max\{p,2\}$ for each  edge $uv$ in $G$.

%   $\multiplicity{D}{x}{y} \le 3k^2$, for any $x,y$ such that $\arc{xy} \in E_d$.
 \end{lemma}
 \begin{proof}
 Suppose that $\multiplicity{D}{u}{v} + \multiplicity{D}{v}{u} > \max\{p,2\}$ for some edge $uv$ in $G$.
 Observe that if $\multiplicity{D}{u}{v} \ge 1$ and $\multiplicity{D}{v}{u} \ge 1$, then by removing one copy of $\arc{uv}$ and one copy of $\arc{vu}$, we obtain a solution to \BCPP{}
on $(G,w,t)$ with weight at most that of $D$ but with fewer arcs.
{(Note that $\multiplicity{D}{u}{v} -1 + \multiplicity{D}{v}{u}  -1> 2-2=0$, and so we still have a solution)}.
 Therefore, we may assume that $\multiplicity{D}{u}{v}  > \max\{p,2\}$ and $\multiplicity{D}{v}{u} = 0$.

 We now show that there must exist a cycle in $D$ containing a copy of $\arc{uv}$.

Modify $D$ by adding a new vertex $x$, with $t(v)$ arcs from $x$ to $v$
for each $v \in V^+_t$, and $-t(v)$ arcs from $v$ to $x$ for each $v \in
V^-_t$. Let $D^*$ be the resulting directed graph. Then observe that $D^*$
is balanced, and therefore $D^*$ has $\multiplicity{D}{u}{v}$ arc-disjoint
cycles, each containing exactly one copy of $\arc{uv}$. At most $p$ of
these cycles can pass through $x$. Therefore there is at least one cycle
containing $\arc{uv}$ which is a cycle in $D$.

 So now let $v = v_1, v_2, \dots, v_l = u$ be a sequence of vertices such that $\multiplicity{D}{v_i}{v_{i+1}}\ge 1$ for each $i \in [l-1]$.
 Replace one copy of each arc $\arc{v_iv_{i+1}}$ with a copy of $\arc{v_{i+1}v_i}$ and remove $2$ copies of $\arc{uv}$.
 Observe that the resulting graph covers every edge of $G$, and the imbalance of each vertex is the same as in $D$.
 Therefore, we have a solution to \BCPP{} on $(G,w,t)$ with weight at most that of $D$ but with fewer arcs.
This contradiction proves the lemma.
 \end{proof}

In order to solve the {\sc BCPP} on a graph $G$, we first add copies of {edges} to $G$ to produce a multigraph $H$, and then find an orientation of $H$ which is a solution to the {\sc BCPP} on $G$. Thus, $H$ is the undirected version of a solution to the {\sc BCPP} on $G$. The lemma below gives a connectivity condition which must be satisfied by any undirected version of a solution. Furthermore, any multigraph that satisfies this condition has an orientation which is a solution to {\sc BCPP}, and such a solution can be found in polynomial time. We will then be able to solve the {\sc BCPP} on $G$ by searching for the minimum weight graph $H$ that satisfies this condition.

 \begin{definition}

   Let $H = (V,E)$ be an undirected multigraph and $t$ a demand function $V\rightarrow \mathbb{Z}$.
 A \emph{$t$-road} is a directed multigraph $T$ with vertex set $V$ such that for each vertex $v \in V$, $d_T^+(v) - d_T^-(v) = t(v)$.
 We say \emph{$H$ has a $t$-road $T$} if there is a subgraph $H'$  of $H$ such that $T$ is an orientation of $H'$.

 \end{definition}

  For an instance $(G,w,t)$ of the {\sc BCPP} with parameter $p$, it may be useful to think of a $t$-road as a set of $p$ arc-disjoint paths from vertices in $V_t^+$ to vertices in $V_t^-$, although a $t$-road does not necessarily have to have such a simple structure.
 {Indeed, a $t$-road may also contain several closed walks. In particular, we note that any solution to the {\sc BCPP} on $(G,w,t)$ is itself a $t$-road.}

   The following lemma and corollary show the relevance of $t$-roads to the {\sc BCPP}.
   Formally an
input of \BCPP{} is a {simple} graph, but to show Corollary \ref{cor:findH} we will abuse this
formality and allow multigraphs.

  \begin{lemma}\label{lem:balancedWithTRoad}
   Let $H$ be an undirected multigraph and let $(H,w,t)$ be an instance of the {\sc BCPP}. Then $(H,w,t)$ has a solution which is an orientation of $H$ {(which is necessarily an optimal solution)} if and only if $H$ has a $t$-road and for every vertex $v$ of $H$, $d_H(v) - t(v)$ is even. {Furthermore, such a solution can be found in polynomial time.}
  \end{lemma}
  \begin{proof}
  Suppose first that  $(H,w,t)$ has a solution of weight $w(H)$. Then there is a directed multigraph $D$ with vertex set $V(H)$ such that $D$ is an orientation of $H$, and $d^+_D(v) - d^-_D(v) = t(v)$ for every vertex $v \in V(H)$. Thus, $D$ itself is a $t$-road which is an orientation of a subgraph of $H$, and so $H$ has a $t$-road. Furthermore, for every vertex $v$ of $H$, $d_H(v) - t(v) = d_D^+(v) + d_D^-(v) - t(v) = d_D^+(v) - d_D^-(v) - t(v) + 2d_D^-(v) = 2d_D^-(v)$, which is even.

  Conversely, suppose that $H$ has a $t$-road and for every vertex $v$ of $H$, $d_H(v) - t(v)$ is even.
   Let $T$ be a $t$-road in $H$. Delete the edges corresponding to $T$ from $H$, and observe that in the remaining graph every vertex $v$ has degree $d_H(v) - d_T^+(v) - d_T^-(v) = d_H(v) - d_T^+(v) + d_T^-(v) - 2d_T^-(v) = d_H(v) - t(v) - 2d_T^-(v)$, which is even.
   Thus in this remaining graph every vertex is of even degree,
%    Thus this remaining graph is balanced,
   and so we may decompose the remaining edges into cycles.
  Orient each of these cycles arbitrarily, and finally add the arcs of $T$. Let $D$ be the resulting digraph.
  Then for each vetex $v \in V(H)$, $d_D^+(v) - d_D^-(v) = 0 + d_T^+(v) - d_T^-(v) = t(v)$.
  Thus $D$ is $t$-balanced and is an orientation of $H$, as required.
  \end{proof}

  By letting $H$ be the undirected version of an optimal solution to an instance $(G,w,t)$, we get the following corollary.

  \begin{corollary}\label{cor:findH}
  Given an instance $(G,w,t)$ of the {\sc BCPP},
  let $H$ be an undirected multigraph of minimum weight, such that {the underlying graph of $H$ is $G$}, $H$ has a $t$-road, and $d_H(v) - t(v)$ is even for every vertex $v$.
  Then there exists an optimal solution to $(G,w,t)$ which is an orientation of $H$, {which can be found in polynomial time.}
  \end{corollary}
% \begin{proof}
%  For an optimal solution $D$ to $(G,w,t)$, let $H$ be the undirected version of $D$. Then the underlying graph of $H$ is $D$, and $D$ is a solution to $(H,w,t)$ of weight $w(H)$, from which it follows by Lemma \ref{lem:balancedWithTRoad} that $H$ has a $t$-road and $d_h(v) - t(v)$ is even for every vertex $v$.
%  Conversely, suppose $H$ has a $t$-road, and $d_H(v) - t(v)$ is even for every vertex $v$. Then by Lemma \ref{lem:balancedWithTRoad} there is a solution $D$ to $(H,w,t)$ which is an orientation of $H$, and $D$ is also a solution to $(G,w,t)$. Furthermore as every odfgsdfgksdbg
% \end{proof}

 Suppose that $G$ has a $t$-road. Then by the above corollary, it is enough to find a minimum weight multigraph $H$ with underlying graph $G$, such that $d_H(v) - t(v)$ is even for every vertex $v$. This can be done in polynomial time as follows.

 Given a graph $G = (V,E)$ and set $X \subseteq V$ of vertices, an \emph{$X$-join} is a set $J \subseteq E$ such that $|J \cup E(v)|$ is odd if and only if $v \in X$, where $E(v)$ is the set of edges incident to $v$.
 Let $X$ be the set of vertices such that $d_H(v) - t(v)$ is odd. Note that if $J$ is an $X$-join { of minimum weight}, the mutigraph $H = (V, E \cup J)$ is a minimum weight multigraph with underlying graph $G$, such that $d_H(v) - t(v)$ is even for every vertex $v$.

 Thus, to solve the {\sc BCPP} on $(G,w,t)$, {where $G$ has a $t$-road}, it is enough to find a minimum weight $X$-join. This problem is known as the {\sc Minimum Weight $X$-Join Problem} (traditionally,  it is called the  {\sc Minimum Weight $T$-Join Problem}, but we use $T$ for $t$-roads) and can be solved in polynomial time:

%  Given an undirected graph $F$ and a set $X$ of vertices of $F$ of even size, a set $J$ of edges of $F$ is an {\em $X$-Join} if $d_{F[J]}(v)$ is odd if and only if $v\in X$.
% When $F$ has weights on its edges, we can speak of the {\sc Minimum Weight $X$-Join Problem}; this problem can be solved in time $O(|V(F)|^3)$ \cite{EdJo1973}.

 \begin{lemma}\label{lem:XJoin}\cite{EdJo1973}
 The {\sc Minimum Weight X-Join Problem} can be solved in time $O(n^3)$.
 \end{lemma}

% A solution to the {\sc Minimum Weight X-Join Problem} can be obtained as follows:
{Let us briefly recall the proof of Lemma \ref{lem:XJoin}.}
Create a graph with vertex set $X$. For any two vertices $u,v \in X$, create an edge $uv$ of weight equal to the minimum weight of a path between $u$ and $v$ in $G$. Find the minimum weight perfect matching in this graph. Then the weight of this matching is the weight of an $X$-join, and an $X$-join can be found by taking the paths corresponding to edges in the matching.

The above remark shows that if $G$ has a $t$-road, then we can solve the {\sc BCPP} in polynomial time.
In general, $G$ may not have a $t$-road.
However, given a solution $D$ to the BCPP on $(G,w,t)$, the undirected
version of $D$ must have a $t$-road (indeed, $D$ itself is a $t$-road).
Therefore if we can correctly guess the part of a solution corresponding to a $t$-road,
 and amend $G$ using this partial solution,
% and add its undirected version to $G$,
the rest of the problem becomes easy.
The following definition and lemmas allow us to restrict such a guess to the places where there are small cuts that prevent a $t$-road from existing.

{
\begin{definition}\label{def:tcut2}
  Let $H=(V,E(H))$ be an undirected multigraph and $t:V \rightarrow \mathbb{Z}$ a demand function such that $\sum_{v\in V}t(v) = 0$.
  %Let $V^+_t=\{v\in V:\ t(v)>0\}$ and $V^-_t=\{v\in V:\ t(v)<0\}$.
   Let $p= \sum_{v \in V^+_t}t(v)$.
  Then a \emph{small $t$-cut} of $H$ is a minimal $(V^+_t,V^-_t)$-cut $F$ such that $|F| < p$.
 \end{definition}
}

A $t$-road {in $H$}, if one exists, can be found in polynomial time by computing a flow of value $p$ from $a$ to $b$ in the unit capacity network $N$ with underlying multigraph
$H^*$, { where $H^*$ is the multigraph derived from $H$ by creating two new vertices $a,b$, with $t(v)$ edges between $a$ and $v$ for each $v \in V^+_t$, and $-t(v)$ edges between $b$ and $v$ for each $v\in V^-_t$}. The next lemma follows from the well-known max-flow-min-cut theorem for $N$.

%The next lemma is proved in the Appendix.
 \begin{lemma}\label{lem:tcutortroad}
 An undirected multigraph $H$ has a $t$-road if and only if $H$ does not have a small $t$-cut.
 \end{lemma}

 The next lemma shows that if we want to decide where to {duplicate} edges to get a $t$-road, we can restrict our attention to the {edges in} small $t$-cuts.

% We say a $t$-road $T$ is \emph{well-behaved} if $\multiplicity{T}{u}{v} + \multiplicity{T}{v}{u} \le 1$ for all $uv \in E \setminus F$.

% \textcolor{blue}{ We say $G$ has a  \emph{well-behaved} $t$-road $T$  if $\multiplicity{T}{u}{v} + \multiplicity{T}{v}{u} \le 1$ for all $uv \in E(G) \setminus F(G)$.}

\begin{definition}
 Let $G=(V,E)$ be an undirected graph and $t:V \rightarrow \mathbb{Z}$ a demand function. 
 Let $F(G)$ be the union of all small $t$-cuts in $G$.
 Then a directed multigraph $T$ is \emph{well-$(G,t)$-behaved} if $\multiplicity{T}{u}{v} = 0$ for all $uv \notin E$ and $\multiplicity{T}{u}{v} + \multiplicity{T}{v}{u} \le 1$ for all $uv \in E \setminus F(G)$.
%  We say a directed multigraph $T$ is \emph{well-$(F,t)$-behaved} if $\multiplicity{T}{u}{v} + \multiplicity{T}{v}{u} \le 1$ for all $uv \in E \setminus F$
\end{definition}

 \begin{lemma}\label{lem:wellbehaved}
  Let $D$ be an optimal solution to \BCPP{} on $(G=(V{(G)},E{(G)}),w,t)$, and 
{ let $H$ be the underlying graph of $D$.}
 Then $H$ has {a well-$(G,t)$-behaved $t$-road}.
 \end{lemma}
\begin{proof}
 Let $F{(G)} \subseteq E{(G)}$ be
%  the set of all edges appearing in a small $t$-cut in $G$.
the union of all small $t$-cuts in $G$.
 Let $J$ be the undirected multigraph derived from $H$ by removing all but one copy of every edge in $E{(G)} \setminus F{(G)}$. {Note that $V(G)=V(H)=V(J)$, $E(G)\subseteq E(J)\subseteq E(H)$, moreover, $H$ and $J$ have 
%  same functions weight and demand functions $w, t$ as $G$.
the same weight function $w$ as $G$.}
Observe that every {(well-$(G,t)$-behaved)} $t$-road in $J$ is also a {(well-$(G,t)$-behaved)} $t$-road in $H$ and every $t$-road in $J$ is {well-$(G,t)$-behaved}.
 So, it is sufficient to show that $J$ has a $t$-road.

Note that if $J$ does not have a $t$-road, then by Lemma~\ref{lem:tcutortroad}, $J$ has a small $t$-cut.
 Note also that by construction, {$D$ is a $t$-road, thus} $H$ has a $t$-road and therefore does not have a small $t$-cut.
 Consider a small $t$-cut $S$ in $J$ and suppose that every edge in $S$ is a copy of an edge in $F{(G)}$.
{ 
As $J$ contains the same number of copies of each edge in $F(G)$ as $H$ does,
it follows that any path from $u \in V^+_t$ to $v \in V^-_t$ in $H \setminus S$
is also a path from $u$ to $v$ in $J \setminus S$. 
But as $H$ does not have a small $t$-cut, such a path must exist, 
contradicting the assumption that $S$ is a small $t$-cut in $J$. }
 Therefore every small $t$-cut in $J$ contains a copy of an edge not in $F{(G)}$. If $J$ has a small $t$-cut, then as {$G$ is a subgraph of $J$,} every small $t$-cut in $J$ is also a small $t$-cut in $G$, it follows that there is a small $t$-cut in $G$ containing edges not in $F$. This is a contradiction by definition of $F$. Therefore we may conclude that $J$ does not have a small $t$-cut, and so $J$ has a $t$-road, as required.
\end{proof}

 If {$|F(G)|$, the number of edges of $G$ in small $t$-cuts,} is bounded by a function on $p$ then, using Lemma \ref{lem:edgeMultiplicity} and Lemma~\ref{lem:wellbehaved}
we can solve BCPP in FPT time by guessing the multiplicities of
each edge in $F$ for an optimal solution $D$.
% (we will show later that for an optimal solution, the multiplicity of each arc may be assumed to be at most $p$).
Unfortunately, $|F(G)|$ may be
larger than any function of $p$ in general. It is also possible to solve the problem on graphs of bounded treewidth
using dynamic programming techniques, but in general the treewidth may be
unbounded.
In Section~\ref{sec:treedecomp} we give a tree decomposition of $G$ in which the
number of edges from $F(G)$ in each bag is bounded by a function of $p$. This
allows us to combine both techniques. In Section~\ref{sec:dp} we give a dynamic
programming algorithm utilizing Lemma~\ref{lem:wellbehaved} that runs in
FPT time.

\section{Tree Decomposition}\label{sec:treedecomp}

\begin{definition}
 Given an undirected graph $G = (V,E)$, a \emph{tree decomposition} of $G$ is a pair $({\cal T}, \beta)$, where ${\cal T}$ is a tree and $\beta:V({\cal T}) \rightarrow 2^V$ such that
 \begin{enumerate}
  \item  $\bigcup_{x \in V({\cal T})}\beta(x) = V$;
  \item  for each edge $uv \in E$, there exists a node $x \in V({\cal T})$ such that $u,v \in \beta(x)$; and
  \item  for each $v \in V$, the set $\beta^{-1}(v)$ of nodes  form a connected subgraph in ${\cal T}$.
 \end{enumerate}
The \emph{width} of $({\cal T}, \beta)$ is $\max_{x \in V({\cal T})}(|\beta(x)|-1)$.
 The \emph{treewidth} of $G$ (denoted $\tw(G)$) is the minimum width of all tree decompositions of $G$.
\end{definition}

 In this section, we provide a tree decomposition of $G$ which we will use for our dynamic programming algorithm.
 The tree decomposition does not have bounded treewidth (i.e. the bags do not have bounded size), but the intersection between bags is small, and each bag has a bounded number of vertices from small $t$-cuts.
 This will turn out to be enough to develop a fixed-parameter algorithm, as in some sense the hardness of \BCPP{} comes from the small $t$-cuts.

 Our tree decomposition is based on a result by Marx, O'Sullivan and Razgon \cite{MarxOSullivanRazgon2011}, in which they show that the minimal small separators of a graph ``live in a part of the graph that has bounded treewidth''\cite{MarxOSullivanRazgon2011}.
% Making this notion more precise, we now give the definition of a \emph{torso graph} and the key result of \cite{MarxOSullivanRazgon2011} that we will use.

 %\begin{definition}\cite{MarxOSullivanRazgon2011}
%  Let $G = (V,E)$ be a graph and $S \subseteq V$. The graph $\torso(G,S)$ has vertex set $S$ and vertices $a,b \in S$ are connected by an edge $ab$ if $ab \in E$ or there is a path $P$ in $G$ connecting $a$ and $b$ whose internal %vertices are not in $S$.
% \end{definition}

 %In the following lemma, given a graph $G$ and vertices $a,b,$ an \emph{$a-b$ separator} is a set $S$ of vertices disjoint from $a$ and $b$ such that $a$ and $b$ are if different connected components of $G - S$.
% Let $\bn(G)$ denote the \emph{bramble number} of  $G$. It is known that $\bn(G) = \tw(G)+1$, where $\tw(G)$ is the treewidth of $G$.

{
\begin{definition}  Let $G$ be a graph and $C \subseteq V(G)$. The graph \textit{torso(G,C)} has vertex set $C$ and vertices
$a, b \in C$ are connected by an edge if $ab \in E(G)$ or there is a path $P$ in $G$ connecting $a$ and $b$ whose internal vertices are not in $C$.
\end{definition}
}

 \begin{lemma}\label{lem:separatorTreewidth}
 \cite[Lemma 2.11]{MarxOSullivanRazgon2011}
  Let $a,b$ be vertices of a graph $G = (V,E)$ and let $l$ be the minimum size of an $(a,b)$-separator.
  For some $e \ge 0$, let $S$ be the union of all minimal $(a,b)$-separators of size at most $l+e$. Then there is an $f(l,e)\cdot(|E|+|V|)$ time algorithm that returns a set $S' \supseteq S$
  disjoint from $\{a,b\}$ such that
  $\tw(\torso(G,S'))\le g(l,e)$, for some functions $f$ and $g$ depending only on $l$ and $e$.
 \end{lemma}

 Marx et al.'s result concerns small separators (i.e. sets of vertices whose removal disconnects a graph), whereas we are interested in small cuts (sets of edges whose removal disconnects a graph).  For this reason, we prove the ``edge version'' of Marx et al.'s result{,  which follows directly from their version.}

% Marx et al. show that for a number of graph separation problems, it is possible to derive a graph closely related to a torso graph, which has the same separators as the original input graph. The separation problem %can then be solved on %this new graph which has bounded treewidth.
 %By contrast, we use the torso graph as a tool to construct a tree decomposition of the original graph, which does not have bounded width, but has enough other structural restrictions to make a dynamic programming %algorithm possible.

%Another difference is that Marx et al. are interested in small separators (i.e. sets of vertices whose removal disconnects a graph), whereas we are interested in small cuts (sets of edges whose removal disconnected a %graph). This %necessitates an extra step in the construction of our graph algorithm, to make sure that all minimal cuts are covered by minimal separators.

 \begin{lemma}\label{lem:cutTreewidth}
  Let $a,b$ be vertices of a graph $G = (V,E)$  and let $l$ be the minimum size of an $(a,b)$-cut.
  For some $e \ge 0$, let $D$ be the union of all minimal $(a,b)$-cuts of size at most $l+e$,
  and let $C = V(D) \setminus \{a,b\}$.
%   and let $C$ be the set of all vertices appearing in edges from $D$.
  Then there is an $f(l,e)\cdot(|E|+|V|)$ time algorithm that returns a set $C' \supseteq C$
  disjoint from $\{a,b\}$ such that
  $\tw(\torso(G,C'))\le g(l,e)$, for some functions $f$ and $g$ depending only on $l$ and $e$.
 \end{lemma}

 \begin{proof}
{We may assume that $ab \notin E$, as otherwise all minimal $(a,b)$-cuts must contain $ab$ and deleting $ab$ from $G$ will not change the set $C$ as it is disjoint from $\{a,b\}.$}

  The main idea is to augment $G$ to produce a graph $G^*$ such that every vertex in $C$ is part of a minimal $(a,b)$-separator in $G^*$. We then apply Lemma \ref{lem:separatorTreewidth} to get a set  {$S' \supseteq S$} and tree decomposition of 
 {$\torso(G^*, S')$} of bounded width, and then use this to produce a set $C' \supseteq C$ and tree decomposition of $\torso(G, C')$ of bounded width.

  % IS THERE A WAY OF FINDING ALL MINIMUM CUTS IN LINEAR TIME?????

  We first produce the graph $G^*$ by subdividing each edge $f$ in $G$ with a new vertex $v_f$.
  Let $S$ be the union of all minimal $(a,b$)-separators in $G^*$ of size at most $l+e$.
  Let $l'$ be the minimum size of an $(a,b)$-separator in $G^*$ (note that $l'$ may be different from $l$).

   Observe that for any minimal $(a,b)$-cut $F$  of size at most $l+e$ in $G$, the set $\{v_f: f \in F\}$ is a minimal $(a,b)$-separator in $G^{*}$. This implies that $l' \le l$. Furthermore, given any edge $f' = uv$ such that $f' \in F$,  assuming $u \notin \{a,b\}$, the set $(\{v_f: f \in F\} \setminus \{v_{f'}\}) \cup \{u\}$ is an $(a,b)$-separator in $G^{*}$ and $\{v_f: f \in F\} \setminus \{v_{f'}\}$ is not. So, $X \cup \{u\}$ is a minimal $(a,b)$-separator in $G^{*}$ for some $X \subseteq \{v_f: f \in F\} \setminus \{v_{f'}\}$.
  Therefore, $u$ is in a minimal $(a,b)$-separator in $G^{*}$ of size less than $l+e$.
  A similar argument holds for $v$.
%   , and by a similar argument so is $v$.
  It follows that $\{u,v\} \setminus \{a,b\} \subseteq S$  for any edge $uv \in D$, and so $C \subseteq S$.

%     We thus have that $C \subset S$, where $S$ is the union of all minimal $(a,b)$-separators of size at most $p-1$ in $G^{**}$.

  Let $e' = (l-l')+e$, so we have $l'+e' = l+e$.
  Now apply Lemma~\ref{lem:separatorTreewidth} to get a set $S'$  disjoint from $\{a,b\}$ such that $S \subseteq S'$, and a tree decomposition $({\cal T}, \beta')$ of $\torso(G^{*},S')$ with treewidth at most $g(l',e')$. As $l' \le l$ and $e' \le l+e$, this treewidth is bounded by a function depending only on $l$ and $e$.

%
%   For each edge $e \in E(G)$ such that $v_e \in S'$, let $h(v_e)$ be an endpoint of $e$, chosen arbitrarily.
%   Now let $C' = (S' \cup \{h(v_e): e \in E(G), v_e \in S'\}) \setminus \{v_e: e  \in E(G)\}$. That is, $C'$ is $S'$ with every $v_e$ replaces with $h(v_e)$. Observe that $C \subseteq S' \cap V(G) \subseteq C'$.
%
%   We produce a tree decomposition of $\torso(G,C')$ as follows.
%   Modify $({\cal T}, \beta')$ by replacing every instance of $v_e$ in a bag with $h(v_e)$, for every edge $e \in E(G)$ such that $v_e \in S'$. Let $({\cal T}, \beta)$ be the resulting decomposition. We now show that $({\cal T}, \beta)$ is indeed a tree decomposition of $\torso(G, C')$.
%
%
%   OR ALTERNATIVELY

  Define a function $h:S' \rightarrow V(G)$ as follows. For each edge $f \in E(G)$ such that $v_f \in S'$,
  if $a$ or $b$ is an endpoint of $f$, set $h(v_f)$ to be the other endpoint, and otherwise
  let $h(v_f)$ be an arbitrary endpoint of $f$.
  For every other $v \in S$, let $h(v)=v$.
  Now let $C' = \{h(v): v \in S'\}$.
  Observe that $C \subseteq S' \cap V(G) \subseteq C'$.

   We produce a tree decomposition of $\torso(G,C')$ as follows.
   Given the tree decomposition $({\cal T}, \beta')$  of $\torso(G^*, S')$, define $\beta: V({\cal T}) \rightarrow C'$ by $\beta(x) = h(\beta'(x)) = \{h(v): v \in \beta'(x)\}$.
   We now show that $({\cal T}, \beta)$ is indeed a tree decomposition of $\torso(G, C')$.

  It follows from construction that  $\bigcup_{x \in V({\cal T})}\beta(x) = C' = V(\torso(G,C'))$.

  Now consider an edge $uw$ in $\torso(G,C')$. We will show that there is an edge $st$ in $\torso(G^*,S')$ with $h(s)=u, h(t)=w$.
  It follows that $s,t \in \beta'(x)$ for some node $x \in V({\cal T})$, and consequently $u,w \in \beta(x)$ for the same node $x$.
  This satisfies the second condition of the tree decomposition.

  As $u,w$ are adjacent in $\torso(G,C')$, there must be a path between them which has no internal vertices in $C'$. By subdividing each edge $f$ in this path with the vertex $v_f$, we get a path $P$ between $u$ and $w$ in $G^*$ which has no internal vertices in $C'$.
  Suppose $P$ contains an internal vertex $v$ with $v \in S'$. Observe that $P$ must also contain $h(v)$ (if $h(v)=v$ then this is obvious, and otherwise $v$ has only two neighbours, both of which must be in $P$ and one of which is $h(v)$). If $h(v) \neq u$ and $h(v) \neq w$, then $h(v)$ is also an internal vertex of $P$, and $P$ has an internal vertex in $C'$, a contradiction.
  Therefore the only internal vertices $v$ of $P$ which are in $S'$ are those for which $h(v)=u$ or $h(v)=w$.

  If $P$ does not have any vertices in $h^{-1}(u)$ (which may happen if $u \notin S'$), then $u$ must have a neighbour $v_f$ with $h(v_f)= u$. Then by adding such a neighbour to $P$, we may assume that $P$ contains at least one vertex in $h^{-1}(u)$. Similarly we may assume $P$ contains at least one vertex in $h^{-1}(w)$.
  By considering the shortest subpath of $P$ containing vertices in both $h^{-1}(u)$ and $h^{-1}(w)$, we have that there is a path in $G^*$ with endpoints $s,t \in S'$, with no internal vertices in $S'$, such that $h(s)=u, h(t)=w$.
  It follows that $s,t$ are adjacent in $\torso(G^*,S')$.

  Now consider $\beta^{-1}(u)$ for some vertex $u \in C'$. We wish to show that $\beta^{-1}(u)$ forms a connected subgraph in ${\cal T}$. As $\beta^{-1}(u) = \bigcup \{ \beta'^{-1}(v): v \in h^{-1}(u)\}$, each $\beta'^{-1}(v)$ forms a connected subgraph in ${\cal T}$, and $\beta'^{-1}(v_1) \cap \beta'^{-1}(v_2) \neq \emptyset$ for adjacent $v_1,v_2$ in $\torso(G^*,S')$, it will be sufficient to show that $h^{-1}(u)$ induces a connected subgraph in $\torso(G^*, S')$.
  If $u \in S'$, then all vertices in $h^{-1}(u)\setminus\{u\}$ are adjacent to $u$ in $\torso(G^*,S^*)$, and therefore $h^{-1}(u)$ induces a graph that contains a star rooted at $u$ as a subgraph.
  On the other hand if $u \notin S'$, then for any $v_1, v_2 \in h^{-1}(u)$, there is a path $v_1 u v_2$ in $G^*$, which contains no internal vertices in $S'$, and so $v_1,v_2$ are adjacent in $\torso(G^*,S')$. Therefore $h^{-1}(u)$ induces a clique in $\torso(G^*,S')$.
  In either case, $h^{-1}(u)$ induces a connected subgraph in $\torso(G^*,S')$.
  This satisfies the third condition of the tree decomposition, which completes the proof that  $({\cal T}, \beta)$ is a tree decomposition of $\torso(G, C')$.

  Finally, note that by construction $\max_{x \in V({\cal T})}(|\beta(x)|-1) \le \max_{x \in V({\cal T})}(|\beta'(x)|-1)$, and so $({\cal T}, \beta)$ has width at most $g(l',e')$, which as previously discussed is bounded by a function depending only on $l$ and $e$.

  It remains to analyse the running time.
  Construction of $G^*$ can be done in linear time as we need to process each edge of $G$ once. $G^*$ has $2|E(G)|$ edges and $|V(G)| + |E(G)|$ vertices, and therefore the algorithm of Lemma \ref{lem:separatorTreewidth} takes time $f(l',e')\cdot(|E(G^*)|+|V(G^*)|) \le f(l,l+e)\cdot(3|E(G)|+|V(G)|)\le 3f(l,l+e)\cdot(|E(G)|+|V(G)|)$.
  Finally, transforming the decomposition  $({\cal T}, \beta')$  into  $({\cal T}, \beta)$  takes time $O(|V({\cal T})|\cdot \max_{x \in V({\cal T})}|\beta(x)| = O(|V({\cal T})|\cdot g(l,l+e)$, and we may assume $|V({\cal T})|$ is linear in $|E(G)|+|V(G)|$ as it took linear time to construct. Therefore the total running time is linear in $|E(G)|+|V(G)|$.

 \end{proof}

 We will now use the treewidth result on torso graphs to construct a tree decomposition of the original graph, in which the width may not be bounded, but the intersection between bags and the number of edges in small cuts in each bag is bounded by a function of $p$.
 In order to make our dynamic programming simpler, it is useful to place further restrictions on the structure of a tree decomposition. The notion of a \emph{nice tree decomposition} is often used in dynamic programming, as it can impose a simple structure and can be found whenever we have a tree decomposition.

\begin{definition}
Given an undirected graph $G = (V,E)$, a \emph{nice tree decomposition} $({\cal T}, \beta)$ is a tree decomposition such that ${\cal T}$ is a rooted tree, and each of the nodes $x \in V({\cal T})$ falls under one of the following classes:
 \begin{itemize}
  \item {\bf $x$ is a Leaf node:} Then $x$ has no children in ${\cal T}$;
  \item {\bf $x$ is an Introduce node:} Then $x$ has a single child $y$ in ${\cal T}$, and there exists a vertex $v \notin \beta(y)$ such that $\beta(x) = \beta(y) \cup \{v\}$;
  \item {\bf $x$ is a Forget node:} Then $x$ has a single child $y$ in ${\cal T}$, and there exists a vertex $v \in \beta(y)$ such that $\beta(x) = \beta(y) \setminus \{v\}$;
  \item {\bf $x$ is a Join node:} Then $x$ has two children $y$ and $z$, and $\beta(x) = \beta(y) = \beta(z)$.
 \end{itemize}
\end{definition}

 (Note that sometimes it is also required that $|\beta(x)|=1$ for every leaf node $x$, but for our purposes we allow $\beta(x)$ to be undbounded.)

 It is well-known that given a tree decomposition of a graph, it can be transformed into a nice tree decomposition of the same width in polynomial time \cite{Kloks1994}.

 \begin{lemma}[Lemma 13.1.3, \cite{Kloks1994}]\label{lem:findNiceTreeDecomp}
  For constant $k$, given a tree decomposition of a graph $G$ of width $k$ and $O(n)$ nodes, where $n$ is the number of vertices of $G$, one can find a nice tree decomposition of $G$ of width $k$ and with at most $4n$ nodes in $O(n)$ time.
 \end{lemma}

 It is also known that a tree decomposition of a graph can be found in fixed-parameter time.

 \begin{lemma}\cite{{Bodlaender1996}}\label{lem:findTreeDecomp}
  There exists an algorithm that, given an $n$-vertex
graph $G$ and integer $k$, runs in time $k^{O(k^3)} \cdot n$ and either constructs a tree
decomposition of $G$ of width at most $k$, or concludes that $G$ has treewidth greater than $k$.
 \end{lemma}

Observe that as the running time in Lemma \ref{lem:findTreeDecomp} is  $k^{O(k^3)} \cdot n$, we may assume the tree decomposition has at most $k^{O(k^3)} \cdot n$ nodes. Then applying Lemma \ref{lem:findNiceTreeDecomp}, we have that for any graph $G$ with treewdith $k$, we can find a nice tree decomposition of $G$ with at most $4|V(G)|$ nodes in time fixed-parameter with respect to $k$.

 Our tree decomposition will be similar but not identical to a nice tree decomposition.
 We are now ready to give our tree decomposition, which is the main result of this section.
%  Lemma~\ref{lem:mainTreeDecomp} is proved in the Appendix.

 We believe this lemma may be useful for other problems in which the ``difficult'' parts of a graph are the small cuts or separators.

  \begin{lemma}\label{lem:smallCutTreeDecomp}
  Let $a,b$ be vertices of a graph $G = (V,E)$ and let $l$ be the minimum size of an $(a,b)$-cut (respectively, let $l$ be the minimum size of an $(a,b)$-separator).
  For some $e \ge 0$, let $D$ be the union of all minimal $(a,b)$-cuts of size at most $l+e$, and let
  $C = V(D) \setminus (a,b)$
%   $C$ be the set of all vertices appearing in edges from $D$
  (respectively, let $C$ be the union of all minimal $(a,b)$-separators of size at most $l+e$).

  Then there is an $f(l,e)\cdot(|E|+|V|)$ time algorithm that returns a set $C'$ disjoint from $\{a,b\}$
  and a (binary) tree decomposition $({\cal T}, \beta)$ of $G$ such that:

  \begin{enumerate}
     \item $C \subseteq C'$;
     \item $\beta(x) \subseteq C'$ for any node $x$ in ${\cal T}$ which is not a leaf node (in particular, the intersection between any two bags {of adjacent nodes of $\cal T$} is contained in $C'$);
%      \item For any nodes $x \neq y$ in  ${\cal T}$, $\beta(x) \cap \beta(y) \subseteq S'$;
     \item For any node $x$ in ${\cal T}$, $|\beta(x) \cap C'| \le g(l,e)$;
     \item $({\cal T}, \beta)$ restricted to $C'$ (i.e. $({\cal T}, \beta')$, where $\beta'(x) = \beta(x) \cap C'$) is a nice tree decomposition;
%      except that the leaf bags may have size larger than $1$.
  \end{enumerate}

\noindent  for some functions $f$ and $g$ depending only on $l$ and $e$.
 \end{lemma}

 \begin{proof}
  If $C$ is the union of all vertices appearing in the set $D$ of all minimal $(a,b)$-cuts of size at most $l+e$, then apply Lemma \ref{lem:cutTreewidth}. If $C$ is the union of all minimal $(a,b)$-separators of size at most $l+e$, then apply Lemma \ref{lem:separatorTreewidth}. In either case, we get a set $C' \supseteq C$
  disjoint from $\{a,b\}$ such that
  $\tw(\torso(G,C'))\le g(l,e)$, for a function $g$ depending only on $l$ and $e$.
  From here on the proof is identical for the two cases.

  Using Lemmas \ref{lem:findNiceTreeDecomp} and \ref{lem:findTreeDecomp}, we may find a nice tree decomposition of $\torso(G,C')$ of width at most $g(l,e)$ in time $f(g(l,e))\cdot(|E|+|V|)$, for some function $f$ depending only on $g(l,e)$.
  Let $({\cal T}', \beta')$ be the resulting tree decomposition of $\torso(G,C')$.

We now add the vertices of $G$ which are not in $C'$ to this decomposition.
Consider any component $X$ of $G-C'$.
Then $N(X) \subseteq C'$.
Furthermore, by definition of $\torso(G, C')$, any pair of vertices in $N(X)$ are adjacent in $\torso(G, C')$. It is well-known that the vertices of a clique in a graph are fully contained in a single bag in any tree decomposition of the graph.
Therefore, $N(X) \subseteq \beta'(x)$ for some node $x$ in ${\cal T}$.

If $x$ is a Leaf node then modify $\beta'(x)$ by adding $X$ to it.
Otherwise, modify $({\cal T}', \beta')$ by inserting (in the edge of $\cal T$ between $x$ and its parent) a new Join node $y$ as the parent of $x$, with another child node $z$ of $y$, such that $\beta'(y) = \beta'(x)$, and $\beta'(z) = \beta'(x) \cup X$.
Thus, $X$ is still added to a Leaf node.

Let $({\cal T}, \beta)$ be the resulting tree decomposition.
As every component of $G - C'$ was added to a bag in a tree decomposition of $\torso(G, C')$, $\bigcup_{x \in V({\cal T})}\beta(x) = V(G)$.
Every edge between vertices in $C'$ is in a bag due to the tree decomposition of $\torso(G, C')$, and for every $v \notin C'$, $N(v)$ is contained in the same bag as $v$. Therefore for every edge $uv$ in $G$, $u$ and $v$ appear in the same bag.
For any vertex $v$, ${\beta}^{-1}(v)$ consists of a single node if $v \notin C'$, and otherwise ${\beta}^{-1}(v)$ is connected by the tree decomposition of $\torso(G, C')$.
Thus, $({\cal T}, \beta)$ is a tree decomposition of $G$.

Furthermore, by construction $\beta(x) \subseteq C'$ for every non-leaf node $x$, $|\beta(x) \cap C'|\le g(l,e)$ for every node $x$, and $({\cal T}, \beta)$ restricted to $C'$ is a nice tree decomposition.
% except for the sizes of the leaf nodes.
 \end{proof}

 We now modify this approach slightly to get the desired tree decomposition when $C$ is the union of all edges in small $t$-cuts.

  \begin{lemma}\label{lem:mainTreeDecomp}
 Let  $(G = (V,E),w,t)$ be an instance of \BCPP,
%   let $V^+ = \{v \in V: t(v) > 0\}$ and $V^- = \{v \in V: t(v)<0\}$ and let $p = \sum_{v \in V^+}t(v)$.
  let $C$ be the non-empty set of vertices appearing in
  edges in
  small $t$-cuts.
  %minimal $(V^+, V^-)$ cuts of size less than $p$.
Then there is an $f(p)\cdot(|E|+|V|)$ time algorithm that returns a set $C'$ and a (binary) tree decomposition $({\cal T}, \beta)$ of $G$ such that:

  \begin{enumerate}
     \item $C \subseteq C'$;
     \item $\beta(x) \subseteq C'$ for any node $x$ in ${\cal T}$ which is not a leaf node (in particular, the intersection between any two bags {of adjacent nodes of $\cal T$} is contained in $C'$);
%      \item For any nodes $x \neq y$ in  ${\cal T}$, $\beta(x) \cap \beta(y) \subseteq S'$;
     \item For any node $x$ in ${\cal T}$, $|\beta(x) \cap C'| \le g(p)$;
     \item $({\cal T}, \beta)$ restricted to $C'$ (i.e. $({\cal T}, \beta')$, where $\beta'(x) = \beta(x) \cap S'$) is a nice tree decomposition;
%      , except that the leaf bags may have size larger than $1$.
  \end{enumerate}

\noindent  for some functions $f$ and $g$ depending only on $p$.
 \end{lemma}

 \begin{proof}

  First construct the multigraph $G^*$ from $G$ by creating two new vertices $a,b$, with $t(v)$ edges between $a$ and $v$ for each $v \in V^+_t$, and $-t(v)$ edges between $b$ and $v$ for each $v\in V^-_t$.
  Then by definition, $C$ is the set of vertices appearing in an edge $e \in E(G)$ such that $e$ is part of a minimal $(a,b)$-cut in $G^*$ of size less than $p$.

    Now apply Lemma~\ref{lem:smallCutTreeDecomp}  to get a set $C'$ disjoint from $\{a,b\}$ such that $C \subseteq C'$ and a tree decomposition of $\torso(G^{*},C')$ with treewidth at most $g(l,e)$, where $l+e = p-1$ and so $g(l,e)$ is bounded by a function of $p$.
%   Let $S' = C' \setminus \{a,b\}$, and note that $C \subseteq S'$.
  It follows from the definition of a torso graph that  $\torso(G^* \setminus \{a,b\}, C')$ is a subgraph of $\torso(G^*,C')\setminus \{a,b\}$  \cite[Lemma 2.6]{MarxOSullivanRazgon2011}, and so we can get a tree decomposition $({\cal T}, \beta)$ of $\torso(G,C')$ by removing $a$ and $b$ from every bag in the tree decomposition of $\torso(G^*,C')$.
As $a,b \notin C'$, the resulting tree decomposition is still a nice tree decomposition when restricted to $C'$.
%   It remains to do a little bookkeeping to ensure that $({\cal T}, \beta)$ restricted to $S'$ is indeed a nice tree decomposition.
%   If $a,b \notin C'$, then we have nothing to do.
%   However if $a \in C'$, then for an Introduce node $x$ that introduced $a$ and has child $y$, these nodes will now be such that $\beta(x)=\beta(y)$, against the requirements of a nice tree decomposition. To address this, we contract the node and its child into a single node.
%   Forget nodes for $a$, and Introduce and Forget nodes for $b$, are handled similarly.
 \end{proof}

 \section{Dynamic Programming}\label{sec:dp}
 
 {
  Let $(G,w,t)$ be an instance of \BCPP{}.
  Let $({\cal T}, \beta)$ be the tree decomposition of $G$ and $C'$ the set of vertices containing all vertices of every small $t$-cut given by Lemma~\ref{lem:mainTreeDecomp}.
%  In the previous section, we gave a tree decomposition for $G$ such that in each bag of the decomposition, the number of edges from a small $t$-cut is bounded.
 In this section we give a dynamic programming algorithm based on this decomposition.
 
 Before describing the algorithm, we give some notation that we use in this section.  
  Let $\alpha(x) = \beta(x) \cap C'$.
  Thus $\beta(x) \cap \beta(y) \subseteq \alpha(x)$ for all nodes $x \neq y$, and $\alpha(x) = \beta(x)$ for every non-leaf $x$.
  Furthermore, for any Join node $x$ with two children $y$ and $z$, we have that $\alpha(x) = \alpha(y) = \alpha(z)$, even if one or both of the children of $x$ is a Leaf node whose bag contains vertices not in $C'$.
  Let $\gamma(x)$ be the union of the bags of all predecessors of $x$ including $x$ itself.
  Thus, if $r$ is the root node of ${\cal T}$, then $\gamma(r) = V(G)$.

   Let $h: V(G) \rightarrow \{\textsc{odd}, \textsc{even}\}$ be the function such that $h(v) = \textsc{odd}$ if $t(v)$ is odd and $h(v) = \textsc{even}$ if $t(v)$ is even. Observe that in the undirected version of any solution to \BCPP{} on $(G,w,t)$, each vertex $v$ will have odd degree if and only if $h(v) = \textsc{odd}$. 
 Thus, $h$ and similar functions will be used to tell us whether a vertex should have odd or even degree.

 To simplify some expressions, we adopt the convention that $\textsc{odd} + \textsc{odd} = \textsc{even}, \textsc{even} + \textsc{even} = \textsc{even}$, and $\textsc{odd} + \textsc{even} = \textsc{odd}$.
 We say a vertex $v$ is \emph{$h$-balanced} if it has odd degree if and only if $h(v) = \textsc{odd}$.
 An undirected multigraph $H$ is \emph{$h$-balanced} if every vertex is $h$-balanced.

 We now give an outline of our algorithm.
 
 By Corollary~\ref{cor:findH}, in order to solve an instance $(G,w,t)$ of the {\sc BCPP}, it is enough to find an undirected multigraph $H$ of minimum weight, such that the underlying graph of $H$ is $G$, $H$ has a $t$-road, and $d_H(v) - t(v)$ is even for every vertex $v$. Our algorithm will therefore focus on solving this problem, rather than finding an optimal multi-orientation of $G$ directly. 
 By Lemma~\ref{lem:wellbehaved}, we may assume $H$ has a {well-$(G,t)$-behaved} $t$-road $T$. Our dynamic programming algorithm will give a way to find $H$ and $T$.
 
 For each node $x$ in ${\cal T}$, our dynamic programming algorithm will calculate a value $\psi(x,H',T',t',h')$, for a particular range of graphs $H'$ and $T'$ and functions $t'$ and $h'$. 
 Informally, $\psi(x,H',T',t',h')$ denotes a potential solution of minimum weight, restricted to $\gamma(x)$.
 Let $H$ denote this restricted solution and $T$ its ``$t$-road'', similarly restricted to $\gamma(x)$.
 The subgraph $H'$ tells us how $H$ should look when restricted to $\alpha(x)$, and similarly $T'$ tells us how $T$ should look when restricted to $\alpha(x)$.
 The function $t'$ tells us what the imbalance of each vertex should be within $T$. Roughly speaking, it is the function for which $T$ is a $t'$-road. (Note that $T$ will not necessarily be a $t$-road itself, as it is only a restriction of a potential $t$-road to $\gamma(x)$.)
 In a similar way, the function $h'$, which maps vertices to either $\textsc{odd}$ or $\textsc{even}$, tells us whether the degree of each vertex within $H$ should be odd or even. (We note that in a full solution, the parity of the degree of each vertex $v$ will be defined by $h(v)$, but as $H$ is only a partial solution it may be that $h'(v) \neq h(v)$ for some $v$.)
 
 More formally, 
   let $x$ be a node of ${\cal T}$, let $H'$ be an undirected multigraph with underlying graph $G[\alpha(x)]$,
 such that $\unDirMultiplicity{H'}{u}{v} \le \max\{p,2\}$ for all edges $uv$.
 Let $T'$ be a directed graph with vertex set $\alpha(x)$, such that $\multiplicity{T'}{u}{v} + \multiplicity{T'}{v}{u} \le \unDirMultiplicity{H'}{u}{v}$ for all edges $uv$.
 Let $t'$ be a function $\alpha(x) \rightarrow [-p,p]$ and
 let $h'$ be a function $\alpha(x) \rightarrow \{\textsc{odd}, \textsc{even}\}$.
 Then let $\psi(x,H',T',t',h')$ be an undirected multigraph $H$ with underlying graph $G[\gamma(x)]$, of minimum weight such that
 \vspace{-1mm}
 \begin{enumerate}
  \item $H[\alpha(x)] = H'$.
  \item $H$ has a {well-$(G,t)$-behaved}  $t^*$-road $T$ such that $T$ restricted to $\alpha(x)$ is $T'$, where $t^*: \gamma(x) \rightarrow [-p,p]$ is the function such that $t^*(v) = t'(v)$ for $v \in \alpha(x)$ and $t^*(v) = t(v)$, otherwise.
  \item $H$ is $h^*$-balanced, where $h^*: \gamma(x) \rightarrow \{\textsc{odd}, \textsc{even}\}$ is the function such that $h^*(v) = h'(v)$ if $v \in \alpha(x)$ and $h^*(v) = h(v)$, otherwise.
 \end{enumerate}
 
 If no such $H$ exists, let $\psi(x,H',T',t',h') = {\sc null}$.
 }

 The following lemma shows that to solve the {\sc BCPP}, it is enough to calculate $\psi(x,H',T',t',h')$ for every choice of $x,H',T',t',h'$.

 \begin{lemma}\label{lem:root}
  Let $r$ be the root node of ${\cal T}$.
  Let $t'$ be $t$ restricted to $\alpha(r)$, and let $h'$ be $h$ restricted to $\alpha(r)$.
%\footnote{Recall that $h(v)=\textsc{odd}$ if and only if $t(v)$ is odd.}
  Let $H'$ and $T'$ be chosen such that the weight of $H = \psi(r,H',T',t',h')$ is minimized.
  Then the weight of $H$ is the weight of an optimal solution to the \BCPP{} on $(G,w,t)$, and given $H$ we may construct an optimal solution  to \BCPP{} on $(G,w,t)$ in polynomial time.
 \end{lemma}
\begin{proof}

Observe that by construction of $t'$ and $h'$, $t^*$ and $h^*$ in the definition of $\psi(r,H',T',t',h')$ are $t$ and $h$, respectively. Also observe that for a graph $H$, $d_H^*(v)-t(v)$ is even for each vertex $v$ if and only if $H$ is $h$-balanced.

Let $D$ be an optimal solution to the \BCPP{} on $(G,w,t)$, and let $H$ be the undirected version of $D$. By Lemma \ref{lem:balancedWithTRoad}, $H$ is $h$-balanced and has a $t$-road. Furthermore by Lemma \ref{lem:wellbehaved}, $H$ has a  {well-$(G,t)$-behaved}  $t$-road $T$. By Lemma \ref{lem:edgeMultiplicity}, we may assume that $\unDirMultiplicity{H}{u}{v} \le \max\{p,2\}$ for each edge $uv$.  $H$ clearly has underlying graph $G = \gamma(r)$. So by letting $H'$ be $H[\alpha(r)]$ and letting $T'$ be $T[\alpha(r)]$, we have that $H$ satisfies all the conditions of $\psi(x,H',T',t',h')$  (except possibly for minimality).

On the other hand, suppose $H$ satisfies all these conditions. Then in particular, $H$ has underlying graph $\gamma(r) = G$, $H$ is $h$-balanced, and $H$ has a $t$-road. It follows by Lemma \ref{lem:balancedWithTRoad} that there exists a solution to the \BCPP{} on $(G,w,t)$ which is an orientation of $H$.

It follows that the minimum weight solution  to the \BCPP{} on $(G,w,t)$ has the same weight as  $H =\psi(r,H',T',t',h')$, when $H'$ and $T'$ are chosen such that the weight of $H$ is minimized.
 \end{proof}

 {  When using dynamic programming algorithms based on tree decompositions, the most commonly used approach is to consider the restriction of possible solutions to each bag, and combine information about the possible restrictions on each bag to construct a full solution. However this approach only works when the tree decomposition is of bounded width, as the number of restrictions to consider on each bag is bounded. In our case, some of the bags in the decomposition may be arbitrarily large, so we cannot consider all possible solutions on a bag. However, we do have that each bag contains a bounded number of vertices from $C'$, where $C'$ contains all vertices that appear in edges in small $t$-cuts. It will turn out to be enough to make a guess based on the edges between vertices in $C'$, after which the rest of the problem can be solved efficiently.}

 Finally, we show how to calculate $\psi(x,H',T',t',h')$ for every choice of $x,H',T',t',h'$.

 \begin{lemma}\label{lem:dpalg}
 $\psi(x,H',T',t',h')$ can be calculated in FPT time, for all choices of $x,H',T',t',h'$.
\end{lemma}
 \begin{proof}
 Consider some node $x$, and assume that we have already calculated $\psi(y,H'',T'',t'',h'')$, for all descendants $y$ of $x$ and all choices of $H'',T'',t'',h''$.
 We consider the possible types of nodes separately.

  {\bf $x$ is a Leaf node:} 
  Consider the multigraph $G_x$ with vertex set $\beta(x) = \gamma(x)$ such that $G_x[\alpha(x)] = H'$, and $G_x$ has exactly one copy of each edge in $G[\beta(x)]$ not contained in $\alpha(x)$. Note that $G_x$ is necessarily a subgraph of $H = \psi(x,H',T',t',h')$. Moreover, as every edge from $G[\beta(x)]$ in a small $t$-cut of $G$ is contained in $\alpha(x)$, any well-$(G,t)$-behaved $t^*$-road in $H$ is also a well-$(G,t)$-behaved $t^*$-road in $G_x$.
  
  It follows that if there exists any $H$ satisfying the first two conditions of $\psi(x,H',T',t',h')$ then $G_x$ has a well-$(G,t)$-behaved $t^*$-road.
  So we first check whether $G_x$ has a $t^*$-road (we note that any $t^*$-road in $G_x$ is well-$(G,t)$-behaved by construction of $G_x$). 
  
  If  $G_x$ has a $t^*$-road, it remains to find a minimum weight (multi)set of edges to add to $G_x$ to make it $h^{*}$-balanced. This can be done by solving the {\sc Minimum Weight $X$-Join Problem}, where $X$ is the set of all vertices in $\beta(x)$ that are not $h^{*}$-balanced in $G_x$. By Lemma \ref{lem:XJoin}, this can be done in polynomial time.

  {\bf $x$ is an Introduce node:}

  Let $y$ be the child node of $x$, and let $v$ be the single vertex in $\beta(x) \setminus \alpha(y)$.
  Then no vertices in $\gamma(x)$ are adjacent with $v$, except for those in $\alpha(x)$.
  In particular for any $H = \psi(x,H',T',t',h')$, the only edges of $H$ incident with $v$ are those in $H'$.
  Thus, if $v$ is not $h'$-balanced in $H'$, then $\psi(x,H',T',t',h') = \textsc{null}$.
  Similarly, if $v$ is not $t'$-balanced in $T'$, then $\psi(x,H',T',t',h') = \textsc{null}$.

  Otherwise, suppose that $H = \psi(x,H',T',t',h')$, and let $H^*$ be $H$ restricted to $\gamma(y)$.
  We now construct the values $H'',T'',t'',h''$ for which $H^*=\psi(y,H'',T'',t'',h'')$ must hold.
  Observe that $H^*[\alpha(y)] = H[\alpha(y)] = H'[\alpha(y)]$. Thus we will set $H''=H'[\alpha(y)]$.
  For the well-$(G,t)$-behaved $t'$-road $T$ that $H$ must have, let $T^*$ be $T$ restricted to $\gamma(y)$, and observe that $T^*[\alpha(y)] = T[\alpha(y)] = T'[\alpha(y)]$. Thus we will set $T''=T'[\alpha(y)]$.
  As $T^*$ is equal to $T$ with the arcs incident to $u$ removed, we have that the imbalance of a vertex $u \in \alpha(y)$ in $T^*$ is $t'(u) - \multiplicity{T'}{u}{v} + \multiplicity{T'}{v}{u}$. (Note also that for $u \in \gamma(y) \setminus \alpha(u)$, the imbalance of $u$ remains unchanged i.e. is still $t(u)$.) Thus we let 
  $t'': \alpha(y) \rightarrow [-p,p]$ be the function such that $t''(u) = t'(u) - \multiplicity{T'}{u}{v} + \multiplicity{T'}{v}{u}$.
  By a similar argument, the parity of the degree of any vertex $u \in \alpha(y)$ in $H^*$ is equal to $h''(u)$, where  $h'': \alpha(y) \rightarrow \{\textsc{odd}, \textsc{even}\}$ is such that if $\unDirMultiplicity{H'}{u}{v}$ is odd then
  $h''(u) = h'(u) + \textsc{odd}$,
%   $h''(u) = \textsc{even}$ if $h'(u) = \textsc{odd}$ and $h''(u) = \textsc{odd}$ if $h'(u) = \textsc{odd}$,
  and otherwise $h''(u)=h'(u)$.
  Thus, we have that $H^*=\psi(y,H'',T'',t'',h'')$ (where the minimality property follows from the fact that any improvement on $H^*$ would give a corresponding improvement on $H$),
  
  Thus, we may set $\psi(x,H',T',t',h')$ to be $\psi(y,H'',T'',t'',h'')$ together with the edges of $H'$ incident with $v$.

  {\bf $x$ is a Forget node:}
  
  Let $y$ be the child node of $x$, and let $v$ be the single vertex in $\alpha(y) \setminus \beta(x)$.
  Note that $\gamma(y)=\gamma(x)$.
  Let $t'': \alpha(y) \rightarrow [-p,p]$ be the function $t^*$ restricted to $\alpha(y)$ (i.e. $t''$ extends $t'$ and assigns $v$ to $t(v)$).
   Let $h'': \alpha(y) \rightarrow \{\textsc{odd},\textsc{even}\}$ be the function $h^*$ restricted to $\alpha(y)$ (i.e. extends $h'$ and assings $v$ to $h(v)$).
  If $H = \psi(x,H',T',t',h')$ then by construction of $t''$ and $h''$, we also have $H = \psi(y,H'',T'',t'',h'')$, for some values of $H''$ and $T''$.
  Note that the possible values of $H''$ are those for which $H''[\alpha(x)] = H'$, and so the only choice is the multiplicity of each edge in $H''$ incident with $v$. By Lemma~\ref{lem:edgeMultiplicity} we may assume the multiplicity of each such edge is at most $p$, and therefore we have at most $(p+1)^{|\alpha(x)|}\leq (p+1)^{g(p)}$ possible values of $H''$. 
  Similarly, we have at most $(p+1)^{|\alpha(x)|}\leq (p+1)^{2g(p)}$ possible values of $T''$.
  
  We therefore may set $\psi(x,H',T',t',h')$ to be the minimum weight $\psi(y,H'',T'',t'',h'')$ over all possible values of $H''$ and $T''$.

  {\bf $x$ is a Join node:}

  Let $y$ and $z$ be the children of $x$, and recall that $\alpha(x) = \alpha(y) = \alpha(z)$, and furthermore $\gamma(x)=\gamma(y) \cup \gamma(z)$ and $\gamma(y) \cap \gamma(z) = \alpha(x)$.
  
  Suppose that $H=\psi(x,H',T',t',h')$, and consider the graphs $H_y = H[\gamma(y)]$ and $H_z=H[\gamma(z)]$.
  
  Observe that the weight of $H$ is equal to $w(H_y)+w(H_z) - w(H')$ (as the only edges contained in both $H_y$ and $H_z$ are those in $H'$).
  Since $\alpha(x)=\alpha(y)$, it must be the case that $H_y = \psi(y,H',T',t'',h'')$ for some choice of $t''$ and $h''$.
  Similarly $H_z = \psi(z,H',T',t''',h''')$ for some choice of $t'''$ and $h'''$. It remains to determine the possible choices of $t'',h'',t''',h'''$.

  Consider a well-$(G,t)$-behaved directed multigraph $T$, and let $T_y = T[\gamma(y)]$ and $T_z = T[\gamma(z)]$. Let ${t^*}''$ be the function such that $T_y$ is a ${t^*}''$-road, and let ${t^*}'''$ be the function such that $T_z$ is a ${t^*}'''$-road.
  For any $v \in \alpha(x)$, the imbalance of $v$ in $T$ is equal to ${t^*}''(v) + {t^*}'''(v)  - \sum_{u \in \alpha(x)}\multiplicity{T'}{v}{u}
  + \sum_{u \in \alpha(x)}\multiplicity{T'}{u}{v}$ (where the last two terms come from the fact that arcs in $T'$ are counted twice in ${t^*}''(v) + {t^*}'''(v)$).
  Thus, $v$ is $t^*$-balanced in $T$ if and only if $t^*(v) = {t^*}''(v) + {t^*}'''(v)  - \sum_{u \in \alpha(x)}\multiplicity{T'}{v}{u}
  + \sum_{u \in \alpha(x)}\multiplicity{T'}{u}{v}$.
  We also note that for $v \in \gamma(y)\setminus \alpha(x)$, $v$ is $t^*$-balanced in $T$ if and only if ${t^*}''(v) = t^*(v)=t(v)$, and for $v \in \gamma(z)\setminus \alpha(x)$, $v$ is $t^*$-balanced in $T$ if and only if ${t^*}'''(v) = t^*(v)=t(v)$.
  
  Let ${h^*}'':\gamma(y) \rightarrow \{\textsc{odd},\textsc{even}\}$ be the function such that $H_y$ is ${h^*}''$-balanced, and let  ${h^*}''':\gamma(z) \rightarrow \{\textsc{odd},\textsc{even}\}$ be the function such that $H_z$ is ${h^*}'''$-balanced
  Then by a similar argument, a vertex $v \in \alpha(x)$ is $h^*$-balanced in $H$ if and only if $h^*(v) = {h^*}''(v) + {h^*}'''(v) + c$, where $c = \textsc{even}$ if $v$ has even degree in $H$, and $c = \textsc{odd}$ otherwise.

  The above implies that $\psi(x,H',T',t',h')$ is the union of $\psi(y,H',T',t'',h'')$ and $\psi(z,H',T',t''',h'')$, where $t'',h'',t''',h'''$ are chosen to minimize the the total weight of $\psi(y,H',T',t'',h'')$ and $\psi(z,H',T',t''',h'')$ and such that
  \begin{enumerate}
   \item $t''(v),t'''(v) \in [-p, p]$ for all $v \in \alpha(x)$;
   \item $t'(v) = t''(v) + t'''(v) - \sum_{u \in \alpha(x)}\multiplicity{T'}{v}{u} + \sum_{u \in \alpha(x)}\multiplicity{T'}{u}{v}$ for all $v \in \alpha(v)$;
   \item $h'(v) = h''(v) + h'''(v)$ if $v$ has even degree in $H$, and $h'(v) = \textsc{odd} + h''(v) + h'''(v)$ otherwise.
  \end{enumerate}
 
   \medskip

  Observe that in the case of a Join node, there is only one possible choice of $t'''$ for each choice of $t''$ and only one possible choice of $h'''$ for each choice of $h''$.
  Therefore there are at most $[2(2p+1)]^{g(p)}$ possible choices for $t'',t''',h'',h'''$.
  Therefore it is possible to calculate $\psi(x,H',T',t',h')$ in fixed-parameter time, as long as
   we have already calculated $\psi(y,H'',T'',t'',h'')$, for all descendants $y$ of $x$ and all choices of $H'',T'',t'',h''$.

   It remains to show that the number of graphs $\psi(x,H',T',t',h')$ to calculate is bounded by a function of $p$ times a polynomial in $|V(G)|$.
   We may assume the number of nodes $x$ in ${\cal T}$ is bounded by $|V(G)|$.
 As $|\alpha(x)|$ is bounded by a function of $p$, and $H'$ has at most $\max\{p,2\}$ edges for each edge within $\alpha(x)$, and $T'$ has at most $\max\{p,2\}$ arcs for each edge within $\alpha(x)$, the number of possible graphs $H'$ and $T'$ is bounded by a function of $p$. Finally, as $|\alpha(x)|$ is bounded by a function of $p$, the number of possible functions $t:\alpha(x) \rightarrow [-p,p]$ and
   $h':\alpha(x) \rightarrow \{\textsc{odd}, \textsc{even}\}$ is also bounded by a function of $p$.
 %\end{proof}
\end{proof}

 Lemmas \ref{lem:root} and \ref{lem:dpalg} imply the following:

 \begin{theorem}\label{thm:GCPP}
  \BCPP{} is fixed-parameter tractable.
 \end{theorem}

 Theorem~\ref{thm:GCPP} and Lemma~\ref{lem:initialReduction} imply the following:

 \begin{theorem}
  \kaCPP{} is fixed-parameter tractable.
 \end{theorem}

\section{Related Open Problem}

% We have solved an open problem in \cite{BeNiSoWe} by showing that MCPP  parameterized by the number of arcs is fixed-parameter tractable.
% To prove this result we reduced MCPP to a generalization of UCPP and applied a very useful lemma of
% Marx {\em et al.} \cite{MarxOSullivanRazgon2011} on treewidth of the torso graph with respect to small separators. Note that our application of the lemma is significantly different from those in \cite{MarxOSullivanRazgon2011}
% and we believe that our approach will be of interest in designing fixed-parameter algorithms for other problems.

 Van Bevern {\em et al.} \cite{BeNiSoWe} mention two other parameterizations of MCPP. One of them is by tw$(G)$. It was proved by Fernandes {\em et al.} \cite{FeLeWa2009} that this parameterisation of MCPP is in XP, but it is unknown whether it is FPT  \cite{BeNiSoWe}. A vertex $v$ of $G$ is called even if the number of arcs and edges incident to $v$ is even.
Edmonds and Johnson \cite{EdJo1973} proved that if all vertices of $G$ are even then MCPP is polynomial time solvable. So, the number of odd (not even) vertices is a natural parameter. It is unknown whether the corresponding parameterization of MCPP is FPT  \cite{BeNiSoWe}.

\vspace{3mm}

\noindent{\bf Acknowledgement}  We are grateful to both referees for several very useful suggestions for improvement of the paper. 
Research of GG was supported by Royal Society Wolfson Research Merit Award. {Research of BS was supported by China Scholarship Council.}

 \newpage


\begin{thebibliography}{10}


\bibitem{BaGu2009} J. Bang-Jensen and G. Gutin, Digraphs: Theory, Algorithms and Applications, 2nd Ed., Springer, 2009.

\bibitem{BeNiSoWe} R. van Bevern, R. Niedermeier, M. Sorge, and M. Weller, Complexity of Arc Rooting Problems. Chapter 2 in
A. Corber{\'a}n and G. Laporte (eds.), Arc Routing: Problems, Methods and Applications, SIAM, Phil., in press.

\bibitem{BeBo1974} E. J. Beltrami and L. D. Bodin. Networks and vehicle routing for municipal
waste collection. Networks, 4(1):65--94, 1974.

\bibitem{Bodlaender1996} Bodlaender, H.L. A linear-time algorithm for finding tree-decompositions of small
treewidth. SIAM J. Computing 25(6), 1305--1317 (1996)


\bibitem{Brucker1981}
P. Brucker. The Chinese postman problem for mixed graphs. Lect. Notes Comput. Sci. 100 (1981) 354--366.

\bibitem{Chr1973} N. Christofides. The optimum traversal of a graph. Omega 1(1973) 719--732.

% \bibitem{DowneyFellows99}
% R.~G.~Downey and M.~R.~Fellows,
% {\em Parameterized Complexity},
% Springer, 1999.

\bibitem{DowneyFellows13}
R.~G.~Downey and M.~R.~Fellows,
{\em Fundamentals of Parameterized Complexity},
Springer, 2013.


\bibitem{EdJo1973} J. Edmonds and E. L. Johnson. Matching, Euler tours and the Chinese postman.
Mathematical Programming, 5 (1973) 88--124.


\bibitem{EiGeLa1995}
H. A. Eiselt, M. Gendreau, and G. Laporte. Arc routing problems. I. The Chinese
postman problem. Oper. Res. 43 (1995) 231--242.

\bibitem{FeLeWa2009} C. G. Fernandes, O. Lee, and Y. Wakabayashi. Minimum cycle cover and
Chinese postman problems on mixed graphs with bounded tree-width. Discrete
Applied Mathematics, 157(2):272--279, 2009.

\bibitem{FlumGrohe06}
J.~Flum and M.~Grohe,
{\em Parameterized Complexity Theory},
Springer, 2006.

% \bibitem{GutinJonesSheng14}
% G.~Gutin, M.~Jones and B.~Sheng,
% Parameterized Complexity of the $k$-Arc Chinese Postman Problem.
% \url{http://arxiv.org/abs/1403.1512}

\bibitem{Kloks1994}
T.~Kloks,
Treewidth: Computations and Approximations.
LNCS, vol 842, Springer, Heidelber (1994)

 \bibitem{MarxOSullivanRazgon2011}
D. Marx, B. O'Sullivan and I. Razgon,
Finding small separators in linear time via treewidth reduction. ACM Trans. Algorithms 9 (2013) article 30.

\bibitem{Minieka1979}
E. Minieka. The Chinese postman problem for mixed networks. Management Sci.
25 (1979/80) 643--648.

\bibitem{Niedermeier06}
R.~Niedermeier,
{\em Invitation to Fixed-Parameter Algorithms},
Oxford UP, 2006.

\bibitem{Papadimitriou1976}
C. H. Papadimitriou. On the complexity of edge traversing.  J. ACM 23 (1976) 544--554.

\bibitem{Peng1989}
 Y. Peng. Approximation algorithms for some postman problems over mixed graphs.
Chinese J. Oper. Res. 8 (1989) 76--80.

 \bibitem{Sor2013}
 M. Sorge, Some Algorithmic Challenges in Arc Routing.
 \newblock Talk at NII Shonan Seminar no.
18, May 2013.

 \bibitem{MartinezThesis}
  F.J. Zaragoza Mart\'{i}nez, Postman Problems on Mixed Graphs.
  PhD thesis, University of Waterloo, 2003.



 \end{thebibliography}
\end{document}